\documentclass[a4paper,11pt]{article}
\usepackage[lmargin=4.cm,rmargin=3cm]{geometry}
\usepackage[utf8]{inputenc}
\usepackage{fouriernc}
\usepackage{authblk}
\usepackage{microtype}
\usepackage{amsmath,mathtools}
\usepackage{amsfonts}
\usepackage{amssymb}
\usepackage{amsthm}
\usepackage{mathrsfs}
\usepackage{bbm}
\usepackage[english]{babel} 
\usepackage[compress]{cite}
\usepackage{eufrak}
\usepackage{yfonts}
\usepackage{bmpsize-base}
\usepackage{setspace}
\usepackage{booktabs}
\usepackage{eso-pic}
\usepackage{shapepar}
\usepackage{enumerate}
\usepackage{enumitem}
\usepackage{braket}
\usepackage{graphicx}
\usepackage{epstopdf}
\usepackage{caption}
\usepackage{float}
\usepackage{subfigure}
\usepackage{bm}
\usepackage{makeidx}
\usepackage{placeins}
\usepackage{listings} 
\usepackage{color} 
\usepackage{titlepic}
\usepackage{hyperref}
\usepackage{bbm}
\usepackage{mathrsfs}	
\usepackage{soul}
\usepackage{color}
\usepackage{scalerel}
\usepackage{framed,lipsum}
\usepackage {fancyhdr}
\newtheorem{definition}{Definition}[section]
\newtheorem{theorem}{Theorem}[section]
\newtheorem{corollary}{Corollary}[theorem]
\newtheorem{lemma}[theorem]{Lemma}
\newtheorem{proposition}[theorem]{Proposition}
\definecolor{gray}{rgb}{92.,92.,92.}
\newcommand {\Cdot}{\vcenter{\hbox{\scalebox{1.6}{.}}}}

\newcommand {\Max}{\text{a.s.}-\max}
\title{Properties of the full replica symmetry breaking free energy functional of the Ising spin glass on random regular graph}

\date{}
\author{Francesco Concetti}
\affil{\textit{Dipartimento di Fisica, Sapienza Università di Roma, Piazzale Aldo Moro 2, I-00185 Rome, Italy}}
\begin{document}

\maketitle
\begin{abstract}
The full replica symmetry breaking free energy of the Ising spin glass on random regular graphs is given by the solutions of two auxiliary variational problems inside a global (physical) variational problem on the order parameter. In this paper, we provide a detailed study of the auxiliary variational problems. We get the self-consistency equation for the auxiliary order parameter and obtain the existence and uniqueness results. We also show that the full replica symmetry breaking free energy functional recovers the discrete replica symmetry breaking solutions, by imposing some proper conditions on the physical order parameter.
\end{abstract}
\section{Introduction}
The solution of the Sherrington and Kirkpatrick model (SK)\cite{SK1,SK2}, obtained by Parisi through the so-called replica symmetry breaking ansatz (RSB)\cite{Par1_0,Par1_1,VPM}f, represents one of the most remarkable milestones in the study of disordered and glassy systems. In the last thirty years, great efforts have been accomplished to extend the results of the fully connected models to a wider and more refined class of spin glasses.

Spin glasses on random sparse graphs (diluted spin glasses) are a class of mean-field models including the notion of nearest neighbors, which is absent in the infinite range case. Unfortunately, diluted spin glass theory turns to be a remarkably complicated challenge; in the last thirty years, little progress has been achieved.

By the \emph{cavity method}, Parisi, Mézard, and Zecchina obtained an extension of the so-called discrete RSB theory for such class of models \cite{ParMezRRG1} and proposed an algorithm that successfully provides the 1-RSB approximation of the solution \cite{ParMezRRG2}.
The manuscript \cite{MyPaper} derives the full-RSB free energy functional, through a martingale representation of the free energy and the order parameter (martingale approach). The method is based on the hierarchical random overlap structure, proposed by Aizenman, Sims and Starr \cite{ASS} and on the representation of hierarchical exchangeable random variables, introduced by Austin and Panchenko \cite{Austin,AustinPanchenko}. Indeed, mean-field spin glass models seem to be characterized by the hierarchical exchangeability of pure states \cite{PanchenkoExchange}.

In \cite{MyPaper}, the free energy is the sum of two variational problems, the so-called \emph{auxiliary variational problems}, inside a global variational problem, the so-called \emph{physical variational problem}. This approach takes inspiration from the variational representation of the Parisi functional, proposed by Chen and Auffinger \cite{ChenAuf}.

The equation of the full-RSB free energy leads to several mathematical issues. Moreover, it is worth noting that the discrete RSB solutions cannot be obtained by the free energy functional proposed in \cite{MyPaper}.

In this paper we propose a modification of the free energy functional, proposed in \cite{MyPaper}, in order to take into account also the discrete RSB solutions, and provides a detailed mathematical study of the auxiliary variational problems.

The next section is a brief review of the main results of \cite{MyPaper}. Section \ref{sec:3} provides a proper formulation of the auxiliary variational problem and fixes some notations. In Section \ref{sec:4}, we derive the self-consistency equation of the order parameter of the auxiliary variational problem. We prove that the solution of such equation provides the solution of the auxiliary variational problem and obtain the uniqueness result. The existence result is the main topic of Section \ref{sec:5}. The results of this last section prove that for a proper choice of the parameters of the free energy, the free energy presented in this manuscript encodes the discrete-RSB theories.

\section{The Full-RSB formula}
\label{fullRSBsec}
In this section, we recall the free energy functional presented in \cite{MyPaper}. We consider a particular time-reparametrization of the order parameters. In such a way, we obtain a complete theory, that takes into account all the discrete RSB solutions. Moreover, we write the free energy in terms of the cavity magnetizations $m$, instead of the cavity fields $h$:
\begin{equation}
m=\tanh(\beta h).
\end{equation}
The full-RSB free energy functional is obtained by considering the cavity magnetizations as Wiener functionals.

Let $\Omega:=C([0,1],\mathbb{R}^{2k})$ be the space of continuous functions $\bm{\omega}:[0,1]\to \mathbb{R}^{2k}$ with the pointwise convergence topology and the Borel $\sigma-$algebra $\mathcal{F}$. Let $\mathbb{W}_{\nu}$ be the probability measure such as the coordinate map process
\begin{equation}
\bm{W}(\bm{\omega}):=\{\bm{W}(q,\bm{\omega})=\bm{\omega}(q);\,q\in[0,1]\}\,,
\end{equation}
together with its natural filtration, is a vector process where the $2\times k$ components are independent Brownian motions, and the starting point $\bm{W}(0,\bm{\omega})=\bm{\omega}(0)$ is a normal distributed point in $\mathbb{R}^{2k}$. In the following, such Brownian motion will be simply indicated by $\bm{\omega}$. We denote by $\nu$ the multivariate normal distribution of a $2k-$dimensional real random vector. The probability space is equipped with the usual augmentation of the natural filtration of the Brownian motion is denoted by $\{\mathcal{F}_q\}_{q\in[0,1]}$.

Let us introduce the random variables $\Psi^{\text{(e)}}:\Omega\to \mathbb{R}$ and $\Psi^{\text{(v)}}:\Omega\to \mathbb{R}$, defined as follows
\begin{equation}
\begin{gathered}
\label{claim_start}
\Psi^{\text{(e)}}(1,\bm{\omega})=\sum^k_{i=1}\log\left(1+J_{1,2}\,Y\,m(1,\omega_i)m(1,\omega_{i+k})\right)\,,\\
\Psi^{\text{(v)}}(1,\bm{\omega})=\sum^1_{i=0}\log\left(\sum_{\sigma\in\{-1,1\}}\prod^k_{j=1}\left(1+J_{0,j}\,Y\,\sigma\,m(1,\omega_{j+k\times i})\right)\,\right)\,.\\
\end{gathered}
\end{equation}
The quantities $J_{1,2}$ and $J_{0,j}$ are ${0,1}-$valued random variables with probability measure $\mathbb{P}$, such as:
\begin{equation}
\mathbb{P}(\{-1\})=\mathbb{P}(\{1\})=\frac{1}{2}
\end{equation}
The symbol $Y$ stands for
\begin{equation}
Y=\tanh(\beta)
\end{equation}
where $\beta$ is the inverse of the temperature. The cavity magnetization $m$ is a real valued functional $m:C([0,1],\mathbb{R})\to \mathbb{R}$, measurable with respect to $\mathcal{F}_1$.
\begin{equation}
\label{adapted_notation}
\begin{gathered}
m(1,\omega_i)=m(1,\{\omega_i(q');0\leq q'\leq 1\}),\\
\Psi^{\text{(e/v)}}(1,\bm{\omega})=\Psi^{\text{(e/v)}}(1,\{\bm{\omega}(q');0\leq q'\leq 1\}),
\end{gathered}
\end{equation}
where the symbol $\Psi^{\text{(e/v)}}$ refers both to $\Psi^{\text{(e)}}$ and $\Psi^{\text{(v)}}$.
We denote by $\mathcal{E}(x \bm{r})$ the following Doléand-Dade exponential (DDE):
\begin{equation}
\label{DDE_xr}
\mathcal{E}(x \bm{r};q,\bm{\omega})\\=\exp\left(\int^{q}_0 x(q) \,\bm{r}(q', \bm{\omega}\,)\cdot \text{d}\bm{\omega}(q')-\frac{1}{2}\int^{q}_0 dq x^2(q) \,\|\bm{r}(q', \bm{\omega}\,)\|^2\right)\,.
\end{equation}
The variable $\bm{r}$ is a vector stochastic process, adapted to the filtration $\{\mathcal{F}_q\}_{q\in[0,1]}$, with $2 c$ components. We use the shorthand notation $(q,\bm{\omega})$ ( or $(q,\omega)$), after a symbol indicating a stochastic process or a random variable, to denote that such quantity depends on $q\in[0,1]$ and on the realization of the Brownian motion $\bm{\omega}(q')$ (or $\omega(q')$ ) at each time $0\leq q' \leq q$:
\begin{equation}
\label{adapted_notation2}
\bm{r}(q,\bm{\omega})=\bm{r}(q,\{\bm{\omega}(q');0\leq q'\leq q\})\,.
\end{equation}

Let us define
\begin{multline}
\label{funcAux}
\Gamma\big(\,\Psi^{\text{(e/v)}},\,x,\,\bm{r};\,0,\,\bm{\omega}(0)\,\big)=\mathbb{E}\left[\mathcal{E}(x \bm{r};1,\bm{\omega})\Psi^{\text{(e/v)}} ( \,1,\bm{\omega}\,)\Big| \,\{\bm{\omega}(0)\}\, \right]\\
-\frac{1}{2}\mathbb{E}\left[\int^1_0 \text{d}q\,x(q)\mathcal{E}(x \bm{r};q,\bm{\omega})\,\|\,\bm{r}(\,q, \,\bm{\omega}\,) \,\|^2\,\,\bigg| \,\{\bm{\omega}(0)\}\,\right]\,,
\end{multline}
where $\mathbb{E}\left[\cdot|\{\bm{\omega}(0)\}\right]$ is the expectation value with respect to the vectorial Brownian motion $\bm{\omega}$, conditionally to a fixed realization of the starting point $\bm{\omega}(0)$.

The auxiliary variational representation is given by
\begin{equation}
\label{GibbsTimeFunctional2}
\phi^{\text{(e/v)}}(\,0,\bm{\omega}{(0)}\,)=
\max_{\bm{r}} \Gamma\left(\,\Psi^{\text{(e/v)}},\,x,\,\bm{r};\,0,\,\bm{\omega}(0)\,\right)\,.
\end{equation}
The maximum in \eqref{GibbsTimeFunctional2} must be obtained by taking the cavity magnetizaton functional $m$ fixed.

The RSB free energy per spin of the Ising spin glass on a random regular graph, presented in \cite{MyPaper}, is given by:
\begin{equation}
f=-\beta \inf_m \left\{\int\text{d}\nu(\,\bm{\omega}(0)\,)\phi^{\text{(v)}}(0,\bm{\omega}(0))-\int\text{d}\nu(\,\bm{\omega}(0)\,)\phi^{\text{(e)}}(0,\bm{\omega}(0))\right\}\,.
\end{equation}
We may guess that the auxiliary variational problem \eqref{GibbsTimeFunctional2} can be solve by imposing a proper stationary condition:
\begin{equation}
\label{stationary0}
\frac{\delta \Gamma\left(\,\Psi^{\text{(e/v)}},\,x,\,\bm{r};\,0,\,\bm{\omega}(0)\,\right)}{\delta r_i(\,q,\bm{\omega}\,)} =0\,.
\end{equation}
In the next section we will provide a detailed mathematical analysis of the auxiliary variational problem.

\section{General formulation of the problem}
\label{sec:3}
In this section we define a generalization of the auxiliary variational problem \eqref{GibbsTimeFunctional2}. In particular, we define a functional with the same form as \eqref{GibbsTimeFunctional2}, depending on a random variables $\Psi$ and on a function $x:[0,1]\to[0,1]$.

Throughout this and the next chapter, the symbol $\bm{\omega}$ denotes a $n-$dimensional Brownian motion ($n\in\mathbb{N}$), starting from a random point $\bm{\omega}(0)$ at $q=0$, defined on a given probability space $(\Omega,\mathcal{F},\nu\times \mathbb{W})$ (it is not the same probability space of the previous chapter). The symbol $\nu$ denotes the probability measure associated to $\bm{\omega}(0)$, while $\mathbb{W}$ is the probability measure associated to $\bm{\omega}-\bm{\omega}(0)$. Let $\{\mathcal{F}_q\}_{q\in[0,1]}$ be the usual augmented natural filtration of $\bm{\omega}$.

We denote by $\mathbb{E}_{\nu}[\Cdot]$ the expectation value with respect the probability measure $\nu\times \mathbb{W}$. We use the short-hand notation $\mathbb{E}[\Cdot]$ for the expectation value with respect the probability measure $\mathbb{W}$, conditionally to a given realization of $\bm{\omega}(0)$:
\begin{equation}
\mathbb{E}[\Cdot]=\mathbb{E}_{\nu}[\Cdot|\mathcal{F}_0]=\mathbb{E}_{\nu}[\Cdot|\{\bm{\omega}(0)\}]\,.
\end{equation}
We denote by $\int \text{d}\nu(\bm{\omega}(0))\,\Cdot$ the average with respect the starting point:
\begin{equation}
\mathbb{E}_{\nu}[\Cdot]=\int \text{d}\nu(\bm{\omega}(0))\mathbb{E}[\Cdot]\,.
\end{equation}
Let us define define
\begin{itemize}
\item $L_1^{\infty}(\Omega)$,the space of $\mathcal{F}_1-$measurable bounded\footnote{A $\mathcal{F}_1-$measurable random variable $X: \Omega\to\mathbb{R}$ is bounded if there exist a constant $M\leq \infty$ such as $|X(\bm{\omega})|\leq M$ with probability $1$} random variables $X: \Omega\to\mathbb{R}$.
\item $L_1^{p}(\Omega)$,the space of $\mathcal{F}_1-$measurable bounded random variables $X: \Omega\to\mathbb{R}$, such as $\mathbb{E}_{\nu}\left[|X(1,\bm{\omega})|^p\right]<\infty$.
\item $H_{[0,1]}^p(\Omega)$, the space of $n-$dimensional adapted processes $\bm{r}: [0,1]\times\Omega\to\mathbb{R}^n$ satisfying $\mathbb{E}_{\nu}\left[\,\left(\int^1_0\text{d}q\,\|\bm{r}(q,\bm{\omega})\|^2\right)^{\frac{p}{2}}\,\right]<\infty$, with $p\geq 1$.
\item $S_{[0,1]}^p(\Omega)$, the space of $n-$dimensional adapted processes $\phi:[0,1]\times \Omega\to\mathbb{R}$ satisfying $\mathbb{E}_{\nu}\Big[\,\underset{q\in[0,1]}{\sup}|\phi(q,\bm{\omega})|^{p}\,\Big]<\infty$.
\end{itemize}
For convenience, we introduce the following notation
\begin{equation}
\label{norms}
\begin{gathered}
X\in L_1^{\infty}(\Omega) \,:\quad \underset{\bm{\omega}\in\Omega}{\Max} |X(1,\bm{\omega})|=\inf \,\left\{M\in \mathbb{R}; |X(1,\bm{\omega})|\leq M\,\,a.s.\right\}\,,\\
\bm{r}\in H_{[0,1]}^p(\Omega) \,:\quad \|\bm{r}\|_{2,p}=\mathbb{E}_{\nu}\left[\,\left(\int^1_0\text{d}q\,\|\bm{r}(q,\bm{\omega})\|^2\right)^{\frac{p}{2}}\,\right]^{\frac{1}{p}}\,,\\
\phi\in S_{[0,1]}^p(\Omega) \,:\quad \|\phi\|_{\infty,p}=\mathbb{E}_{\nu}\Big[\,\underset{q\in[0,1]}{\sup}|\phi(q,\bm{\omega})|^{p}\,\Big]^{\frac{1}{p}}\,.
\end{gathered}
\end{equation}
We denote by $\chi$ the set of increasing deterministic function $x:[0,1]\to[0,1]$:
\begin{equation}
\chi :=\left\{x:[0,1]\to[0,1];\,\, ,x(0)=0\,,\,\,\,0<x(q)\leq x(q')\,\forall\,0<q\leq q' \leq 1\,\right\}.
\end{equation}
Let us endow the set $\chi$ with the uniform norm $\|x\|_{\infty}=\sup_{q\in[0,1]}\,x(q)$.

Given a process $\bm{r}\in H_{[0,1]}^p(\Omega)$, with $p\geq 2$, a function $x\in\chi$ and a number $q'\in[0,1]$, let $\zeta(\bm{r},x|q')$ be the process defined by:
\begin{multline}
\label{zeta}
\zeta(\bm{r},x; q,\bm{\omega}|q')\\=
\int^q_{q'}x(q') \bm{r}(q'\,, \,\bm{\omega}) \cdot \text{d}\bm{\omega}(q') -\frac{1}{2}\int^q_{q'}dq'\,x^2(q')\,\|\bm{r}(q'\,, \,\bm{\omega})\|^2\,,\quad\text{if}\,\quad q\geq q'
\end{multline}
and
\begin{equation}
\zeta(\bm{r},x; q,\bm{\omega}|q')=0\,,\quad \text{if}\,\quad q<q'\,,
\end{equation}
As in \eqref{DDE_xr}, let $\mathcal{E}(x \bm{r})$ and $\mathcal{E}(x \bm{r}|q')$ be the DDE defined as
\begin{equation}
\begin{gathered}
\mathcal{E}(x \bm{r})=e^{\zeta(\bm{r},x)}\to \mathcal{E}(x \bm{r};q,\bm{\omega})=e^{\zeta(\bm{r},x;q,\bm{\omega})}\,,\\
\mathcal{E}(x \bm{r}|q')=e^{\zeta(\bm{r},x|q')}\to \mathcal{E}(x \bm{r};q,\bm{\omega}|q')=e^{\zeta(\bm{r},x;q,\bm{\omega}|q')}\,.
\end{gathered}
\end{equation}
The symbol $x\bm{r}\in$ denotes the process taking values $x(q)\bm{r}(q,\bm{\omega})$.

In the following, given any symbols $K$ and $\alpha$ and a number $q'\in[0,1]$, the notation $K(\alpha|q')$ refers to an adapted process with a functional dependency on a parameter $\alpha$ and the number $q'$, while $K(\alpha;q,\bm{\omega}|q')=K(\alpha;q,\{\bm{\omega}(q'');0\leq q''\leq q\}|q')$ is the a value of the process at a given "time" $q\in [0,1]$ and a given realization of the Brownian motion $\bm{\omega}$; we also set $K(\alpha)=K(\alpha|0)$.

Now, let us define the set $\widehat{D}_{[0,1]}(\Omega)\subset H_{[0,1]}^p(\Omega)$, with $p\geq 1$, where each element $\bm{r}\in \widehat{D}_{[0,1]}(\Omega)$ identifies a constant $C_{\bm{r}}\geq 0$ such as
\begin{equation}
\label{boundAux}
\left| \zeta(\bm{r},x; 1,\bm{\omega}) \right|\leq C_{\bm{r}}\quad a.s.\,.
\end{equation}
As far as we know, the set $\widehat{D}_{[0,1]}(\Omega)$ is not a vector space.

Note that, if $\bm{r}\in \widehat{D}_{[0,1]}(\Omega)$, then the DDE $\mathcal{E}(x \bm{r})$ is a true martingale. This property assures that:
\begin{equation}
\label{martingality}
\mathbb{E}[\mathcal{E}(x \bm{r};q,\bm{\omega}|q')|\mathcal{F}_{q''}]=1,\quad \forall \,\,\, 0\leq q''\leq q'\leq q\leq 1\,\,\,\text{and}\,\,\,\bm{r}\in \widehat{D}_{[0,1]}(\mathbb{R}^d)\,.
\end{equation}
and
\begin{equation}
\label{positivity}
\mathcal{E}(x \bm{r};q,\bm{\omega}|q')\geq e^{-C_{\bm{r}}}\quad a.s.\,.
\end{equation}
For this reason, the martingale $\mathcal{E}(x \bm{r})$ can be considered as a probability density function of a probability measure $\widetilde{\mathbb{W}}_{x \bm{r}}$ equivalent\footnote{Two probability measures $\widetilde{\mathbb{W}}$ and $\mathbb{W}$, defined on the same measurable space $(\Omega,\mathcal{F})$, are equivalent if, given any set $A\in\mathcal{F}$, then $\mathbb{W}[A]=0$ if and only if $\widetilde{\mathbb{W}}[A]=0$} to $\mathbb{W}$.

Given two processes $\bm{r}$ and $\bm{v}$ in $\widehat{D}_{[0,1]}(\Omega)$, let us introduce the binary functional $\mathbb{D}_{KL}(\, \cdot \, \|\,\cdot \,):\widehat{D}_{[0,1]}(\Omega)\times \widehat{D}_{[0,1]}(\Omega)\to [0,\infty)$ defined by:
\begin{multline}
\mathbb{D}_{KL}(\bm{r} \, \|\,\bm{v} \,)=\mathbb{E}_{\nu}\left[\mathcal{E}(x \bm{r};q,\bm{\omega})\log\left(\frac{\mathcal{E}(x \bm{r};q,\bm{\omega})}{\mathcal{E}(x \bm{v};q,\bm{\omega})}\right)\right]\\=\mathbb{E}_{\nu}\left[\mathcal{E}(x \bm{r};q,\bm{\omega})\int^1_0 dq\,x^2(q)\|\bm{r}(q,\bm{\omega})-\bm{v}(q,\bm{\omega})\|^2\right]
\end{multline}
Because of the property \eqref{boundAux}, the above quantity is defined for all pair of processes in $\widehat{D}_{[0,1]}(\Omega)\times \widehat{D}_{[0,1]}(\Omega)$. Let also define
\begin{equation}
\label{metric}
\mathbb{D}^{\text{sym}}_{KL}(\bm{r} \, \|\,\bm{v} \,)=\max\big\{\mathbb{D}_{KL}(\bm{r} \, \|\,\bm{v} \,),\mathbb{D}_{KL}(\bm{v} \, \|\,\bm{r} \,)\big\}
\end{equation}
Note that the quantity $\mathbb{D}_{KL}(\bm{r} \, \|\,\bm{v} \,)$ and $\mathbb{D}_{KL}(\bm{v} \, \|\,\bm{r} \,)$ are actually the \emph{relative entropies} (or Kullback–Leibler divergences) between the two probability-densities/DDEs $\mathcal{E}(x \bm{v})$ and $\mathcal{E}(x \bm{r})$. Indeed, for any $\bm{r} \in \widehat{D}_{[0,1]}(\Omega)$, one finds $\mathbb{D}_{K,L}(\bm{r} \, \|\,\bm{r} \,)=0$.

Moreover, if two processes $\bm{r}$ and $\bm{v}$ in $\widehat{D}_{[0,1]}(\Omega)$ verify the relation
\begin{equation}
\label{eqRel}
\mathbb{D}^{\text{sym}}_{KL}(\bm{r} \, \|\,\bm{v} \,)=0\,,
\end{equation}
then the two corresponding DDEs are two statistical equivalent densities, i.e for each $\mathcal{F}-$measurable bounded random variable $A:\Omega\to \mathbb{R}$, they verify the equivalence:
\begin{equation}
\label{MyEquality}
\mathbb{E}_{\nu}\left[\,\mathcal{E}(x \bm{r})A(1,\bm{\omega})\right]=\mathbb{E}_{\nu}\left[\,\mathcal{E}(x \bm{v})A(1,\bm{\omega})\right].
\end{equation}
The relations \eqref{eqRel} and \eqref{MyEquality} are equivalence relations.

By the relation \eqref{MyEquality}, we may argue that if \eqref{eqRel} holds, then the processes $\bm{r}$ and $\bm{v}$ are "similar", in some sense. This observation justifies the introduction of the quotient set $D_{[0,1]}(\Omega)$.
\begin{definition}
\label{define_D}
Given a process $\bm{r}\in\widehat{D}_{[0,1]}(\Omega)$, let $[\bm{r}]$ be the equivalence class
\begin{equation}
[\bm{r}]:=\{\bm{v}\in\widehat{D}_{[0,1]}(\Omega);\,\,\mathbb{D}^{\text{sym}}_{KL}(\bm{r} \, \|\,\bm{v} \,)=0\}\,.
\end{equation}
We denote by $D_{[0,1]}(\Omega)$ the set of the equivalence classes:
\begin{equation}
D_{[0,1]}(\Omega):=\{[\bm{r}];\,\,\bm{r}\in\widehat{D}_{[0,1]}(\Omega)\}\,.
\end{equation}
\end{definition}
By abuse of notation, henceforth we will omit the square bracket around the elements of the space $D_{[0,1]}(\Omega)$.

Now, we formulate the variational problems that we aim to study in this paper.
\begin{definition}
\label{definitions}
For a given $n\in\mathbb{N}$, the call RSB value process the functional $\Gamma^{(n)}: L_1^{\infty}(\Omega)\times \chi \times D_{[0,1]}(\Omega)\to S_{[0,1]}^p(\mathbb{R})$ defined as follows:
\begin{multline}
\label{value_function}
\Gamma^{(n)}(\Psi,x,\bm{r};q,\bm{\omega})=\mathbb{E}\left[\mathcal{E}(x \bm{r};1,\bm{\omega}|q)\Psi ( \,1,\bm{\omega}\,)\big|\mathcal{F}_q \right]\\
-\frac{1}{2}\mathbb{E}\left[\int^1_q dq'\,x(q')\mathcal{E}(x \bm{r};q',\bm{\omega}|q)\,\|\bm{r}(q', \,\bm{\omega}) \,\|^2\Bigg|\mathcal{F}_q\,\right]\,,
\end{multline}
where
\begin{itemize}
\item the random variable $\Psi\in L_1^{\infty}(\Omega)$ is the claim;
\item the function $x\in\chi$ is the POP (Parisi Order Parameter);
\item the process $\bm{r}\in D_{[0,1]}(\Omega)$ is the control parameter.
\end{itemize}
We say that a pair $(\Psi,x)\in L_1^{\infty}(\Omega)\times \chi$ allows the RSB expectation if there exists a solution pair $(\phi^{(n)}(\Psi,x),\bm{r}(\Psi,x))\in S_{[0,1]}^p(\Omega)\times D_{[0,1]}(\Omega)$ such that:
\begin{equation}
\phi^{(n)}(\Psi,x;q,\bm{\omega})=\Gamma^{(n)}(\Psi,x,\bm{r}(\Psi,x);q,\bm{\omega})=\sup_{\bm{r}\in D_{[0,1]}(\Omega)}\Gamma^{(n)}(\Psi,x,\bm{r};q,\bm{\omega})\,.
\end{equation}
and the following quantity
\begin{equation}
\label{auxiliary_variational_problem}
\Sigma^{(n)}(\Psi,x)=\int \text{d}\nu(\bm{\omega}(0))\phi^{(n)}(\Psi,x;0,\bm{\omega}(0))=\int \text{d}\nu(\bm{\omega}(0))\sup_{\bm{r}\in D_{[0,1]}(\Omega)}\Gamma^{(n)}(\Psi,x,\bm{r};0,\bm{\omega}(0) )
\end{equation}
is the RSB expectation of $\Psi$, driven by $x$.
\end{definition}
For the rest of the chapter we will omit the superscript $\cdot^{(n)}$ and we consider a generic dimension $n$. We consider only bounded claims because of the fact that, at non-zero temperature, the random variable $\Psi^{\text{(e)}}$ and $\Psi^{\text{(v)}}$ , defined in \eqref{claim_start}, are bounded.

Throughout the chapter, we consider a real constant $c< \infty$ and assume that the claim $\Psi$ is bounded by:
\begin{equation}
\label{boundedness}
|\Psi(1,\bm{\omega})|\leq c\,,\quad a.s. \,.
\end{equation}

The aim of the next section is obtaining a self-consistency equation for the solution pair.

\section{Backward Stochastic Differential equations}
\label{sec:4}
In this section we compute the variation of RSB$-$value process with respect the control parameter.

In the first subsection, we remind some properties of the Doléan-Dade exponential. In the second subsection, we provide a proper definition of the stationary condition \eqref{stationary0}, and we get an equation for the control parameter.
\subsection{Properties of the Doléans-Dade exponential (DDE)}
\label{subsec6.2.1}
In this subsection, we remind some fundamental facts about the DDEs and fix some notations. The DDEs play a crucial role in stochastic theory, and a vast literature has been produced about (see for example Chapter VIII of \cite{YoRev}).

The Doéans-Dade exponential (DDE) is defined as the unique strong solution of the following stochastic differential equation \cite{YoRev,Oksendal}:
\begin{equation}
\label{DDE_derivative}
\mathcal{E}(x\bm{r};q',\bm{\omega}|q)=1+\int^{q'}_{q}\mathcal{E}(x\bm{r};q'',\bm{\omega}|q') x(q)\bm{r}(q'',\bm{\omega})\cdot d\bm{\omega}(q)\,,\quad 0\leq q\leq q'\leq 1\,.
\end{equation}
For any given process $\bm{r}\in D_{[0,1]}(\Omega)$, the DDE $\mathcal{E}(x\bm{r})$ defined in \eqref{DDE_xr} is a true martingale; the martingale condition implies that DDE $\mathcal{E}(x\bm{r})$ is a positive process with
\begin{equation}
\label{MartiProp}
\mathbb{E}[\mathcal{E}(x\bm{r};q,\bm{\omega})]=1\quad \forall q\in[0,1].
\end{equation}
Then, we can define a probability measure $\widetilde{\mathbb{W}}_{x \bm{r}}$ on the measurable space $(\Omega,\mathcal{F})$, equivalent to the Wiener measure $\mathbb{W}$ and such as the DDE $\mathcal{E}(x\bm{r})$ is the Radon-Nikodym derivative of $\widetilde{\mathbb{W}}_{x \bm{r}}$ with respect to $\mathbb{W}$ \cite{YoRev,Billingsley}:
\begin{equation}
\frac{d\widetilde{\mathbb{W}}_{x \bm{r}} }{d\mathbb{W}}(\bm{\omega})=\mathcal{E}\left(\,x\bm{r};1,\bm{\omega}\,\right)\,.
\end{equation}

\begin{definition}
\label{DDE_expectation}
Let $\bm{r}\in D_{[0,1]}(\Omega)$ and $x\in \chi$. For any $\mathcal{F}_1-$measurable random variable $A$, the expectation of $A$ with respect the probability measure $\widetilde{\mathbb{W}}_{x \bm{r}}$ is the linear functional $A\mapsto \widetilde{\mathbb{E}}_{x\bm{r}}[A]\in \mathbb{R}$, defined as follows:
\begin{equation}
\widetilde{\mathbb{E}}_{x\bm{r}}\left[\,A(\,\bm{\omega}\,)\right]=\mathbb{E}\left[\mathcal{E}\left(\,x\bm{r};1,\bm{\omega}\,\right)\,A(\,\bm{\omega}\,)\,\right].
\end{equation}
Moreover, by Bayes Theorem, the conditional expectation value is given by
\begin{equation}
\label{conditionalNot}
\widetilde{\mathbb{E}}_{x\bm{r}}\left[\,A(\,\bm{\omega}\,)\big|\mathcal{F}_q\right]=\mathbb{E}\left[\mathcal{E}(x\bm{r};1,\bm{\omega}|q)\,A(\,\bm{\omega}\,)\,\big|\mathcal{F}_q\right],\quad \forall q\in[0,1]\,.
\end{equation}
\end{definition}
For any process $\bm{r}\in D_{[0,1]}(\Omega)$, we define the vector semimartingale $\bm{W}_{x\bm{r}}$ such as:
\begin{equation}
\label{semimart}
\bm{W}_{x\bm{r}}\left(\,q,\bm{\omega}\right)=\bm{\omega}(q)-\int^q_0 dq'\, x(q')\, \bm{r}(\bm{\omega},\,q'\,).
\end{equation}
or equivalently:
\begin{equation}
\label{AuxProb}
d\bm{W}_{x\bm{r}}\left(\,q,\bm{\omega}\right)=\text{d}\bm{\omega}(q)- \text{d}q\, x\, \bm{r}\left(\bm{\omega},\,x\,\right).
\end{equation}
By Cameron-Martin-Girsanov Theorem (CMG)\cite{CameronMartin,Girsanov}, the vector semimartingale $\bm{W}_{x\bm{r}}$ is a vector Brownian motion with respect the probability measure $\widetilde{\mathbb{W}}_{x \bm{r}}$ and the filtration $\{\mathcal{F}_q\}_{q\in[0,1]}$.

As a consequence, the stochastic integral of any vector processes $\bm{u} \in H^p_{[0,1]}(\Omega)$ (for any $p\geq 1$) with respect the process $\bm{W}_{x\bm{r}}$, is a $\{\mathcal{F}_q\}_{q\in[0,1]}-$martingale with respect $\widetilde{\mathbb{W}}_{x \bm{r}}$:
\begin{equation}
\label{MartiPropMod}
\widetilde{\mathbb{E}}_{x\bm{r}}\left[\,\int^{q_1}_{0} \,\bm{u}(\,q\,, \,\bm{\omega}\,) \cdot d\bm{W}_{x\bm{r}}\left(\,q,\bm{\omega}\right)\Bigg|\mathcal{F}_{q_2}\right]=
\begin{cases}
\int^{q_1}_{0} \,\bm{u}(\,q\,, \,\bm{\omega}\,) \cdot d\bm{W}_{x\bm{r}}\left(\,q,\bm{\omega}\right),\, \text{if}\, q_1\leq q_2,\\
\int^{q_2}_{0} \,\bm{u}(\,q\,, \,\bm{\omega}\,) \cdot d\bm{W}_{x\bm{r}}\left(\,q,\bm{\omega}\right),\, \text{if}\,q_1>q_2,
\end{cases}
\end{equation}
and
\begin{equation}
\label{MartiPropMod1}
\widetilde{\mathbb{E}}_{x\bm{r}}\left[\,\int^{1}_0 \,\bm{u}(\,q\,, \,\bm{\omega}\,) \cdot d\bm{W}_{x\bm{r}}\left(\,q,\bm{\omega}\right)\,\}\right]=0.
\end{equation}
Moreover, the expectation value of the product between two stochastic integrals verify the It\^o isometry
\begin{multline}
\label{MartiPropMod2}
\widetilde{\mathbb{E}}_{x\bm{r}}\left[\,\int^{1}_0 \,\bm{u}(q,\bm{\omega}) \cdot d\bm{W}_{x\bm{r}}\left(\,q,\bm{\omega}\right)\int^{1}_0 \,\bm{v}(q,\bm{\omega}) \cdot d\bm{W}_{x\bm{r}}\left(\,q,\bm{\omega}\right)\,\}\right]=\\
\widetilde{\mathbb{E}}_{x\bm{r}}\left[\,\int^{1}_0 \text{d}q\,\bm{u}(q,\bm{\omega}) \cdot \bm{v}(q,\bm{\omega})\,\}\right],
\end{multline}
for any pair of processes $\bm{u}$ and $\bm{v}$ in $H_{[0,1]}^p(\Omega)$, with $p\geq 2$.

By combining \eqref{MartiPropMod1} with the definition \eqref{semimart}, we also have:
\begin{equation}
\label{integration_stoc}
\widetilde{\mathbb{E}}_{x\bm{r}}\left[\,\int^{1}_0 \,\bm{u}(\,q\,, \,\bm{\omega}\,) \cdot d\bm{\omega}(q)\right]=\widetilde{\mathbb{E}}_{x\bm{r}}\left[\,\int^{1}_0 \text{d}q\,x(q)\bm{u}(q,\bm{\omega}) \cdot \bm{r}(q,\bm{\omega})\right]
\end{equation}
We end the subsection by providing the following notation. Let $\bm{r}$ and $\bm{v}$ be two processes in $D_{[0,1]}(\Omega)$, we define
\begin{equation}
\mathcal{E}(x \bm{v};q',\bm{W}_{x \bm{r}}|q)=\exp\left(\int^{q'}_q \,\bm{v}(q,\bm{\omega}) \cdot d\bm{W}_{x\bm{r}}(q,\bm{\omega})-\frac{1}{2}\int^{q'}_q \text{d}q\,x(q)\|\bm{v}(q,\bm{\omega}) \|^2\right)\,.
\end{equation}
The DDE $\mathcal{E}(x \bm{v})$ is a martingale with respect the probability measure $\widetilde{\mathbb{W}}_{x \bm{r}} $ and verify all the above relation, by substituting $\mathbb{E}\to\widetilde{\mathbb{E}}_{x\bm{r}}$ and $\bm{\omega}(q)\to \bm{W}_{x\bm{r}}(q,\bm{\omega})$.

Now, we have all the necessary tools to address the study of the RSB-expectation.
\subsection{The stationary condition}
Now, we provide a proper definition of "stationary condition". Throughout the subsection, the claim $\Psi$ and the POP $x$ are kept fixed.

Given two vector processes $\bm{r}\in D_{[0,1]}(\Omega)$ and $\bm{u}\in D_{[0,1]}(\Omega)$ and a number $\epsilon \in [0,1]$, let
\begin{equation}
\Theta^{\epsilon}\left(\,\Psi,\,x,\,\bm{r},\bm{u}; q,\bm{\omega}\,\right)=\epsilon \Gamma\left(\,\Psi,x,\bm{r};\,q,\bm{\omega}\,\right)+(1-\epsilon)\Gamma\left(\,\Psi,x,\bm{u};\,q,\bm{\omega}\,\right)
\end{equation}
and
\begin{equation}
\rho^{\epsilon}(\Psi,\,x,\,\bm{r},\bm{u}; q,\bm{\omega})=(1-\epsilon)\, \mathcal{E}(\,x\bm{r};q,\bm{\omega})+\epsilon\mathcal{E}(\,x\bm{u};q,\bm{\omega})\,.
\end{equation}
Note that, by \eqref{martingality} and \eqref{positivity}, since $0\leq \epsilon\leq 1$, then $\rho(\Psi,\,x,\,\bm{r},\bm{u})$ is a strictly positive and bounded martingale of mean $1$. As a consequence, there exists a process $\bm{v}^{\epsilon}(\bm{r},\bm{u})\in D_{[0,1]}(\Omega)$ such as:
\begin{equation}
\rho^{\epsilon}(\Psi,\,x,\,\bm{r},\bm{u}; q,\bm{\omega})=\mathcal{E}(\,x\bm{v}^{\epsilon}(\bm{r},\bm{u});q,\bm{\omega})\,,
\end{equation}
and
\begin{equation}
\Theta^{\epsilon}\left(\,\Psi,\,x,\,\bm{r},\bm{u}; q,\bm{\omega}\,\right)= \Gamma\left(\,\Psi,x,\bm{v}^{\epsilon}(\bm{r},\bm{u});\,q,\bm{\omega}\,\right)\,.
\end{equation}
A straightforward computation yields:
\begin{equation}
\bm{v}^{\epsilon}(\bm{r},\bm{u};q,\bm{\omega})=\frac{(1-\epsilon)\,\bm{r}(q,\bm{\omega}) \mathcal{E}(\,x\bm{r};q,\bm{\omega})+\epsilon\,\bm{u}(q,\bm{\omega})\mathcal{E}(\,x\bm{u};q,\bm{\omega})}{(1-\epsilon)\, \mathcal{E}(\,x\bm{r};q,\bm{\omega})+\epsilon\,\mathcal{E}(\,x\bm{u};q,\bm{\omega})}\,.
\end{equation}
Let
\begin{equation}
\label{direction}
\delta \bm{u}(q,\bm{\omega})=\partial_{\epsilon}\bm{v}^{\epsilon}(\bm{r},\bm{u};q,\bm{\omega})\big|_{\epsilon=0}\,,
\end{equation}
where the symbol $\partial_{\epsilon}$ denote the derivative over $\epsilon$ by taking all the other parameters fixed.

We define the directional derivative of the functional $\bm{r}\mapsto \Gamma(\Psi,x,\bm{r})$ along the path $\{\bm{r},\bm{u}\}$ by
\begin{multline}
\label{GatDer}
\begin{aligned}
&\Pi\left(\,\Psi,\,x,\,\bm{r},\delta \bm{u}; q,\bm{\omega}\,\right)=\partial_{\epsilon}\Gamma\left(\,\Psi,x,\bm{v}^{\epsilon}(\bm{r},\bm{u});\,q,\bm{\omega}\,\right)\big|_{\epsilon=0}\\
&=\widetilde{\mathbb{E}}_{x\bm{r}}\left[\,\Psi ( \,1,\bm{\omega}\,)\int^1_q x(q')\,\delta \bm{u}(q',\bm{\omega}) \cdot d\bm{W}_{x\bm{r}}(q',\bm{\omega})\,\,\Bigg|\mathcal{F}_q \right]
\end{aligned}\\
-\frac{1}{2}\,\,\widetilde{\mathbb{E}}_{x\bm{r}}\left[\int^1_q \text{d}q'\,x(q')\,\int^{q'}_q x(q'')\,\delta \bm{u}(q'',\bm{\omega}) \cdot d\bm{W}_{x\bm{r}}(q'',\bm{\omega})\,\,\left \|\,\bm{r}(q', \bm{\omega}) \,\right\|^2\,\Bigg|\mathcal{F}_q\right]\\
-\,\widetilde{\mathbb{E}}_{x\bm{r}}\left[\int^1_q \text{d}q'\,x(q') \,\delta \bm{u}(q', \,\bm{\omega})\cdot\bm{r}(q', \bm{\omega}) \Bigg|\mathcal{F}_q\right].
\end{multline}
In the following, the process $x\delta \bm{u}$ will be called \emph{direction}.

We guess that the set of all the possible directions $x\delta \bm{u}$, defined as in \eqref{direction}, is actually is actually dense on $H_{[0,1]}^p(\Omega)$. For this reason, we give the following definition of \emph{stationary point}.
\begin{definition}[Stationary point]
A stationary point $\bm{r}^{*}$ of the functional $\Gamma(\Psi,x,\Cdot)$ is a vector stochastic process in $D_{[0,1]}(\Omega)$, such as the directional derivative \eqref{GatDer} vanishes for any direction $x\,\delta\bm{u}\in H_{[0,1]}^p(\Omega)$:
\begin{equation}
\label{statCondPi}
\Pi\left(\,\Psi,\,x,\,\bm{r}^*,\delta \bm{u}; q,\bm{\omega}\,\right)=0\,,\quad \forall\,x\,\delta\bm{u}\in H^p_{[0,1]}(\Omega)\,.
\end{equation}
The above equation is the stationary condition.
\end{definition}
Such definition will be justified a posteriori. We want to obtain an equation for the auxiliary order parameter that is equivalent to the above stationary condition and does not depend on the derivative direction $\delta\bm{u}$.

By relation \eqref{MartiPropMod}, the second expectation value in \eqref{GatDer} can be rewritten in such a way
\begin{multline}
\frac{1}{2}\,\,\widetilde{\mathbb{E}}_{x\bm{r}}\Bigg[\,\int^1_q \text{d}q'\,x(q')\int^{q'}_q x(q'')\,\bm{u}(q'',\bm{\omega}) \cdot d\bm{W}_{x\bm{r}}(q'',\bm{\omega})\|\,\bm{r}(q', \,\bm{\omega}) \,\|^2\Bigg|\mathcal{F}_q\Bigg]=\\
\frac{1}{2}\widetilde{\mathbb{E}}_{x\bm{r}}\left[\,\int^1_q x(q')\,\bm{u}(q',\bm{\omega}) \cdot d\bm{W}_{x\bm{r}}\left(\,q,\bm{\omega}\right)\,\,\int^1_q \text{d}q'\,x(q')\, \|\,\bm{r}(q', \,\bm{\omega}) \|^2\Bigg|\mathcal{F}_q\right]
\end{multline}
and, by formula \eqref{MartiPropMod2}, the third expectation value leads to
\begin{multline}
\widetilde{\mathbb{E}}_{x\bm{r}}\left[\int^1_q \text{d}q'\,x(q') \,\delta \bm{u}(q', \,\bm{\omega})\cdot\bm{r}(q', \,\bm{\omega}) \Bigg|\mathcal{F}_q\right]\\=\widetilde{\mathbb{E}}_{x\bm{r}}\left[\int^1_q \,x(q') \,\delta \bm{u}(q', \,\bm{\omega})\cdot d\bm{W}_{x\bm{r}}(q',\bm{\omega})\int^1_q \bm{r}(q', \,\bm{\omega}) \cdot d\bm{W}_{x\bm{r}}(q',\bm{\omega})\Bigg|\mathcal{F}_q\right]\,.
\end{multline}
Combining the above formulas in \eqref{GatDer}, we can rewrite the directional derivative \eqref{GatDer} in such a way:
\begin{equation}
\label{NewGat}
\Pi\left(\,\Psi,\,x,\,\bm{r},\delta\bm{u}; q,\bm{\omega}\,\right)
=\widetilde{\mathbb{E}}_{x\bm{r}}\left[\pi\left(\,\Psi,\,x,\,\bm{r}; 1,\bm{\omega}\,|q\right)\int^1_q x(q')\,\bm{u}(q',\bm{\omega}) \cdot d\bm{W}_{x\bm{r}}(q',\bm{\omega})\Bigg|\mathcal{F}_q\right]\,,
\end{equation}
with
\begin{multline}
\label{piii}
\pi\left(\,\Psi,\,x,\,\bm{r}; 1,\bm{\omega}\,|q\right)\\=\Psi ( \,1,\bm{\omega}\,)-\int^{1}_q \,\bm{r}(\,q'\,, \bm{\omega}\,) \cdot d\bm{W}_{x\bm{r}}(q',\bm{\omega})-\frac{1}{2}\int^1_q \text{d}q'\,x(q')\left \|\,\bm{r}(q', \,\bm{\omega}) \,\right\|^2.
\end{multline}
In the following we will refer to this quantity as \emph{random RSB}.

A process $\bm{r}^{*}$ is a stationary point, according to the definition \eqref{statCondPi}, if and only if the random RSB $\pi\left(\,\Psi,\,x,\,\bm{r}|q\right)$ is uncorrelated, under the measure $\widetilde{\mathbb{W}}_{x \bm{r}}$, to all the random variables of the form
\begin{equation}
\label{A_direction}
A(1,\bm{\omega}|q\,)=\int^1_q x(q')\,\delta\bm{u}(q',\bm{\omega}) \cdot d\bm{W}_{x\bm{r}}(q',\bm{\omega}),\quad\text{with}\quad x\delta\bm{u}\in H^p_{[0,1]}(\Omega).
\end{equation}
This condition provides the stationary equation for the control parametr $\bm{r}$. The following theorem is one of the most important results of the paper.

\begin{theorem}
\label{stationariTheo}
Let us consider a claim $\Psi \in L_1^{\infty}(\Omega)$ and a POP $x\in \chi$. A control parameter $\bm{r}^{*}\in D_{[0,1]}(\Omega)$ verifies the stationary condition \eqref{statCondPi} if and only if, at any time $q\in[0,1]$, the random RSB $\pi\left(\,\Psi,\,x,\,\bm{r}^*|q\right)$ is $\mathcal{F}_q-$measurable.

In particular, this implies that there exists a process $\phi:[0,1]\times \Omega \to \mathbb{R}$, adapted to the filtration $\{\mathcal{F}_{q'}\}_{q'\in[0,1]}$, such as
\begin{equation}
\pi\left(\,\Psi,\,x,\,\bm{r}; 1,\bm{\omega}\,|q\right)=\phi(q,\bm{\omega})\,,
\end{equation}
so we find the equation:
\begin{equation}
\label{selfEq1Aux}
\phi(q,\bm{\omega})=\Psi ( \,1,\bm{\omega}\,)-\int^{1}_q \,\bm{r}^{*}(q'\,, \bm{\omega}\,) \cdot d\bm{\omega}(q')+\frac{1}{2}\int^1_q \text{d}q'\,x(q') \|\,\bm{r}^{*}(q', \,\bm{\omega}) \|^2.
\end{equation}
The above equation is the stationary equation that generate the RSB expectation.
\end{theorem}
We remind that the random RSB is $\mathcal{F}_q-$measurable if, given a realization of the vector Brownian motion $\bm{\omega}$, $\pi\left(\,\Psi,\,x,\,\bm{r};1,\bm{\omega}|q\right)$ depends only on $\{\bm{\omega}(q'), \,0\leq q' \leq q\,\}$.
The equation \eqref{selfEq1Aux} is the stationary equation of the control parameter.

Note that all the components of $\bm{r}^{*}$ and the process $\phi$ are unknowns of the equation. However, such class of equations may have a unique solution, since the condition that both the process $\phi$ and $\bm{r}$ are adapted provides a further constraint. The right-hand member of the equation, indeed, is a sum of random quantities that are $\mathcal{F}_1$ measurable. We must look for a control parameter $\bm{r}^{*}$, depending only on the past, such as to "delete the dependence of the future".

The stationary equation can be rewritten in stochastic differential notation as
\begin{equation}
d\phi(q,\bm{\omega})=d\bm{\omega}(q)\cdot\bm{r}^{*}(q\,, \bm{\omega}\,) -\frac{1}{2}\text{d}q\,x(q) \|\,\bm{r}^{*}(q, \,\bm{\omega}) \|^2
\end{equation}
together with the end point condition
\begin{equation}
\phi(1,\bm{\omega})=\Psi(1,\bm{\omega}).
\end{equation}
This kind of equation are called \emph{backward stochastic differential equation} (BSDE) \cite{PaPeng}.

BSDEs arise in many optimization and control problems, where the aim to fulfill a given "claim" (the claim $\Psi(1,\bm{\omega})$) and the control parameters depend only on the past(so we consider only adapted process).

\begin{proof}[Proof of Theorem \ref{stationariTheo}]
If the control parameter $\bm{r}^{*}\in D_{[0,1]}(\Omega)$ verifies the condition of \eqref{selfEq1Aux}, then there exists an adapted process $\phi$ such as
\begin{multline}
\Pi(\Psi,\,x,\,\bm{r}^{*},\delta\bm{u};q,\bm{\omega})\\
=\widetilde{\mathbb{E}}_{x\bm{r}^{*}}\left[\,\pi(\Psi,\,x,\,\bm{r}^{*};1,\bm{\omega}|q)\,\int^1_0 x(q)\delta \bm{u}(q,\bm{\omega}) \cdot d\bm{W}_{x\bm{r}^{*}}\left(\,q,\bm{\omega}\right)\Bigg|\mathcal{F}_q\right]\\
=\phi(q,\bm{\omega})\widetilde{\mathbb{E}}_{x\bm{r}^{*}}\left[\,\int^1_0 x(q)\delta \bm{u}(q,\bm{\omega}) \cdot d\bm{W}_{x\bm{r}^{*}}\left(\,q,\bm{\omega}\right)\Bigg|\mathcal{F}_q\right].
\end{multline}
The process $\phi$ can be pulled outside the conditional expectation value, since it is $\mathcal{F}_q-$measurable. The stationarity \eqref{statCondPi} follows from the martingale property \eqref{MartiPropMod1}.

Conversely, suppose that the process $\bm{r}^{*}$ is a stationary point, according to the definition \eqref{statCondPi}. We have to show that the set of the random variables of the form \eqref{A_direction} is dense in $L_1^p(\Omega)$.

Consider the process $\bm{v}\in D_{[0,1]}(\Omega)$, defined as
\begin{equation}
\bm{v}(q,\bm{\omega})=\bm{f}(q) -x(q)\bm{r}^{*}(q,\bm{\omega})\,,
\end{equation}
where $\bm{f}:[0,1]\to \mathbb{R}^m$ is any deterministic and function such as $\int^1_0\text{d}q \,\|\bm{f}(q)\|^p=1$, with $p\geq 2$ . Using standard notation \cite{Berry}, we wright that $\bm{f}\in L^p([0,1],\mathbb{R}^n)$.

Since $x\, \delta \bm{u}\in H_{[0,1]}^p(\Omega)$, and $x(q)>0$ for all $q>0$, then we can consider a class of directions of the form:
\begin{equation}
\label{protoTypeder}
x(q') \delta \bm{u}(q',\bm{\omega})=\bm{v}(q',\bm{\omega})\mathcal{E}\big( \bm{v};q',\bm{W}_{x \bm{r}^{*}}\big|q\big)\,,\quad \text{with}\quad q'\geq q
\end{equation}
Then by property \eqref{DDE_derivative} of DDEs, one get
\begin{equation}
\int^1_q x(q')\delta \bm{u}(q',\bm{\omega}) \cdot d\bm{W}_{x\bm{r}^{*}}(q',\bm{\omega})\\
=\frac{\mathcal{E}(\bm{f};q',\bm{\omega}|q)}{\mathcal{E}(x\bm{r}^{*};q',\bm{\omega}|q)}-1\,.
\end{equation}
Replacing the above direction in \eqref{NewGat}, the stationary condition \eqref{statCondPi} yields
\begin{multline}
\mathbb{E}\left[\mathcal{E}(\bm{f};q',\bm{\omega}|q)\pi(\Psi,\,x,\,\bm{r}^{*};1,\bm{\omega}|q)\,\big|\mathcal{F}_q\right]\\
=\widetilde{\mathbb{E}}_{x\bm{r}^{*}}\left[\,\pi(\Psi,\,x,\,\bm{r}^{*};1,\bm{\omega}|q)\,\big|\mathcal{F}_q\right]\,,\,\forall\,\bm{f}\in L^p([0,1],\mathbb{R}^n).
\end{multline}
Note that the expectation value on the left-hand side member of the equation is with respect the probability measure $\mathbb{W}$ and on the right-hand the expectation is with respect $\widetilde{\mathbb{W}}_{x \bm{r}^{*}}$.

The linear span of the set $\left\{\mathcal{E}(\bm{f};q',\bm{\omega}|q),\bm{f} \in L^p([0,1],\mathbb{R}^n) \right\}$ is dense in $L_1^p(\Omega)$ (Lemma 4.3.2. in \cite{Oksendal}), so the above equation implies:
\begin{equation}
\label{in_the_proof}
\pi(\Psi,\,x,\,\bm{r}^{*};1,\bm{\omega}|q)
=\widetilde{\mathbb{E}}_{x\bm{r}^{*}}\left[\,\pi(\Psi,\,x,\,\bm{r}^{*};1,\bm{\omega}|q)\,|\mathcal{F}_q\right]\,\quad a.s..
\end{equation}
By the definition of conditional expectation, the right member in the above equation is an $\mathcal{F}_q-$measurable random variable, so we may consider an adapted process $\phi$, such as
\begin{equation}
\label{MeaningOfC}
\phi(q,\bm{\omega} )=\pi(\Psi,\,x,\,\bm{r}^{*};1,\bm{\omega}|q)\,.
\end{equation}
that conclude the proof.
\end{proof}
The meaning of the process $\phi$ is stated in the following
\begin{corollary}
\label{GlobalMinimumC}
Given a claim $\Psi \in L_1^{\infty}(\Omega)$ and a POP $x\in \chi$, if the pair of processes $(\phi,\bm{r}^*)\in S_{[0,1]}^p(\Omega)\times D_{[0,1]}(\Omega)$ is a solution of the BSDE \eqref{selfEq1Aux}, then
\begin{equation}
\label{value_process_solution}
\phi(q,\bm{\omega})=\Gamma(\Psi,x,\bm{r}^*;q,\bm{\omega})\,.
\end{equation}
\end{corollary}
\begin{proof}
Since $\phi(q,\Cdot)$ is $\mathcal{F}_q-$measurable, then $\phi(q,\bm{\omega})=\widetilde{\mathbb{E}}_{x\bm{r}^{*}}[\phi(q,\bm{\omega})|\mathcal{F}_q]$; then the proof is given by replacing $\phi$ with \eqref{selfEq1Aux} and using the relation \eqref{integration_stoc}.
\end{proof}
If the solution of the stationary condition has a unique solution and provides the global maximum of the RSB value process, then the BSDE \eqref{selfEq1Aux} determines completely the RSB expectation. We will discuss this matter in the next subsection.

\subsection{Global maximum condition}
In this subsection, we prove that the solution of the stationary condition provides the global maximum of the RSB value function. We also discuss some property of the so-called RSB expectation, that we defined in \eqref{auxiliary_variational_problem}.

First of all, we need to state the following result.
\begin{theorem}
\label{existence_uniqueness}
For any give claim and POP $(\Psi,x)\in L_1^{\infty}(\Omega)\times \chi $, there exist a unique pair of processes $(\,\phi(\Psi,x),\bm{r}(\Psi,x)\,)\in S_{[0,1]}^p(\Omega)\times D_{[0,1]}(\Omega)$ that is a soluion of the BSDE \eqref{selfEq1Aux}.
\end{theorem}
We will devote the next chapter to the proof of the existence result. The uniqueness is discussed in this subsection.

We proceed in the same way as in the proof of Corollary \ref{GlobalMinimumC}. Let $(\phi,\bm{r}^*)$ be a solution of \eqref{selfEq1Aux} corresponding to a claim $\Psi$ and a POP $x$. Since $\phi(q,\Cdot)$ is a $\mathcal{F}_q-$measurable random variable, the conditional expectation $\widetilde{\mathbb{E}}_{x\bm{v}}[\Cdot|\mathcal{F}_q]$ of both side in the equation \eqref{selfEq1Aux}, for any $\bm{v}\in D_{[0,1]}(\Omega)$, yields
\begin{multline}
\label{StraightfordComputation}
\phi(q,\bm{\omega})=\widetilde{\mathbb{E}}_{x\bm{v}}\left[\pi(\Psi,\,x,\,\bm{r}^{*};\,1,\bm{\omega}|q)\,|\mathcal{F}_q\right]\\=\widetilde{\mathbb{E}}_{x\bm{v}}\left[\Psi( \,1,\bm{\omega}\,)|\mathcal{F}_q\right]+\frac{1}{2}\widetilde{\mathbb{E}}_{x\bm{v}}\left[ \int \text{d}q x(q)\,\bm{r}^{*}(q, \,\bm{\omega})\,\cdot\left(\bm{r}^{*}(q, \,\bm{\omega})-2 \bm{v}(q, \,\bm{\omega})\right)\,\Bigg|\mathcal{F}_q\right].
\end{multline}
In the rest of the thesis, we will prefer to use a notation that explicitates the dependence of $\phi$ and $\bm{r}^*$ on the calim and the POP.The solution of the BSDE \eqref{selfEq1Aux}, for a given pair $(\Psi,x)\in L_1^{\infty}\times \chi$ will be denoted by $(\,\phi(\Psi,x),\bm{r}(\Psi,x)\,)$.

The following Lemma is an immediate consequence of the above identity.
\begin{lemma}
\label{lemma_Global_minimum}
Let us consider a claim $\Psi$ and a POP $x$ and the corresponding BSDE solution $(\,\phi(\Psi,x),\bm{r}(\Psi,x)\,)$. For any $\bm{v}\in D_{[0,1]}(\Omega)$ the RSB value process $\Gamma$ verifies:
\begin{multline}
\Gamma(\Psi,x,\bm{r}(\Psi,x);q,\bm{\omega})\\=\Gamma(\Psi,x,\bm{v};q,\bm{\omega})+\frac{1}{2}\mathbb{E}_{x \bm{v}}\left[ \int \text{d}q x(q)\,\|\bm{r}(\Psi,x;q, \,\bm{\omega})- \bm{v}(q, \,\bm{\omega})\|^2\Bigg|\mathcal{F}_q\right]
\end{multline}
\end{lemma}

\begin{proof}
Let us consider two processes $\bm{r}$ and $\bm{v}$ in $ D_{[0,1]}(\Omega)$. By definition \ref{definitions}, the RSB value process $\Gamma(\Psi,x,\bm{v})$ is given by
\begin{equation}
\begin{aligned}
&\Gamma(\Psi,\,x,\,\bm{v};\,q,\bm{\omega})\\&=\widetilde{\mathbb{E}}_{x\bm{v}}\left[\Psi ( \,1,\bm{\omega}\,)\,|\mathcal{F}_q\right]-\frac{1}{2} \widetilde{\mathbb{E}}_{x\bm{v}}\left[\int \text{d}q x(q)\,\left \|\,\bm{v}(q, \,\bm{\omega}) \,\right\|^2\,\Bigg|\mathcal{F}_q\right]\\
&\begin{multlined}
=\widetilde{\mathbb{E}}_{x\bm{v}}\left[\Psi( \,1,\bm{\omega}\,)\,|\mathcal{F}_q\right]\\+\frac{1}{2} \widetilde{\mathbb{E}}_{x\bm{v}}\left[\int \text{d}q x(q)\,\,\bm{r}(q, \,\bm{\omega})\,\cdot\left(\bm{r}(q, \,\bm{\omega})-2 \bm{v}(q, \,\bm{\omega})\right)\,\Bigg|\mathcal{F}_q\right]\\-\frac{1}{2} \widetilde{\mathbb{E}}_{x\bm{v}}\left[\int \text{d}q x(q)\,\left \|\,\bm{v}(q, \,\bm{\omega}) -\bm{r}(q, \,\bm{\omega})\,\right\|^2\,\Bigg|\mathcal{F}_q\right]\,.
\end{multlined}
\end{aligned}
\end{equation}
Then, if $\bm{r}=\bm{r}(\Psi,x)$, then, by replacing the identity \eqref{StraightfordComputation} in the above formula, we get
\begin{equation}
\Gamma(\Psi,\,x,\,\bm{v};\,q,\bm{\omega})=\phi(\Psi,x;q,\bm{\omega})-\frac{1}{2} \widetilde{\mathbb{E}}_{x\bm{v}}\left[\int \text{d}q x(q)\,\left \|\,\bm{v}(q, \,\bm{\omega}) -\bm{r}(\Psi,x;q, \,\bm{\omega})\,\right\|^2\,\Bigg|\mathcal{F}_q\right]
\end{equation}
so, using Corollary \eqref{MeaningOfC}, we end the prove.
\end{proof}
From the above lemma, we obtain the most remarkable result of this section.
\begin{theorem}
\label{GlobalMinTheo}
Given the pair $(\Psi,x)\in L_1^{\infty}(\Omega)\times \chi$, let $(\,\phi(\Psi,x),\bm{r}(\Psi,x)\,)$ be the solution of the BSDE \eqref{selfEq1Aux}, then
\begin{equation}
\phi(\Psi,x;q,\bm{\omega})=\underset{\bm{r}\in D_{[0,1]}(\Omega)}{\max}\Gamma(\Psi,x,\bm{r};q,\bm{\omega})\,.
\end{equation}
\end{theorem}
An other important consequence of Lemma \ref{lemma_Global_minimum}, is the uniqueness result.
\begin{proof}[Proof of Theorem \ref{existence_uniqueness}: uniqueness]
For any pair of processes $\bm{r}_1$ and $\bm{r}_2$ in $D_{[0,1]}(\Omega)$, let
\begin{equation}
\mathbb{D}(\bm{r}_1,\bm{r}_2;q,\bm{\omega})=\widetilde{\mathbb{E}}_{x\bm{r}_1}\left[\int \text{d}q x(q)\,\left \|\,\bm{r}_1(q, \,\bm{\omega}) -\bm{r}_2(q, \,\bm{\omega})\,\right\|^2\,\Bigg|\mathcal{F}_q\right]
\end{equation}
The above quantity is obviously non-negative.

Assume, by contradiction, that there exist two distinct pairs $(\phi_1,\bm{r}_1)$ and $(\phi_2,\bm{r}_2)$, on $ S_{[0,1]}^p(\Omega)\times D_{[0,1]}(\Omega)$, that are solutions of the BSDE \eqref{selfEq1Aux}. Applying Lemma \ref{lemma_Global_minimum} for both, we get
\begin{equation}
\Gamma\left(\Psi,x,\bm{r}_2|\,\bm{\omega}(0)\,\right)=\Gamma\left(\Psi,x,\bm{r}_1|\,\bm{\omega}(0)\,\right)-\frac{1}{2}\mathbb{D}(\bm{r}_2,\bm{r}_1;q,\bm{\omega})
\end{equation}
and
\begin{equation}
\Gamma\left(\Psi,x,\bm{r}_1|\,\bm{\omega}(0)\,\right)=\Gamma\left(\Psi,x,\bm{r}_2|\,\bm{\omega}(0)\,\right)-\frac{1}{2}\mathbb{D}(\bm{r}_1,\bm{r}_2;q,\bm{\omega}).
\end{equation}
After some straightforward manipulations, we get
\begin{equation}
\mathbb{D}(\bm{r}_1,\bm{r}_2;q,\bm{\omega})=-\mathbb{D}(\bm{r}_1,\bm{r}_2;q,\bm{\omega}).
\end{equation}
Since both $\mathbb{D}(\bm{r}_2,\bm{r}_1;q,\bm{\omega})$ and $\mathbb{D}(\bm{r}_1,\bm{r}_2;q,\bm{\omega})$ are non-negative, then the above relation implies
\begin{equation}
\mathbb{D}(\bm{r}_1,\bm{r}_2;q,\bm{\omega})=\mathbb{D}(\bm{r}_2,\bm{r}_1;q,\bm{\omega})=0.
\end{equation}
Moreover it is easy to show that:
\begin{equation}
0\leq \mathbb{D}_{KL}(\bm{r}_1\|\bm{r}_2)\leq\widetilde{\mathbb{E}}_{x\bm{r}_1}\left[\mathbb{D}(\bm{r}_1,\bm{r}_2;q,\bm{\omega})\right]\,.
\end{equation}
Then, by the definition of the vector space $ D_{[0,1]}(\Omega)$ \ref{define_D}, if the comparing function between two process $\bm{r}_1$ and $\bm{r}_2$ vanishes, than the two processes are equivalent and correspond to the same element of $ D_{[0,1]}(\Omega)$.
\end{proof}

We end the section by providing some property of the solution $(\phi(\Psi,x),\bm{r}(\Psi,x))$ associate to $(\Psi,x)$. The following statement are direct consequences of the maximum principle \ref{GlobalMinTheo} and the uniqueness in \ref{existence_uniqueness}.
\begin{proposition}
\label{Non:llinear_prop}
The solution of the BSDE \eqref{selfEq1Aux} verifies
\begin{itemize}
\item Let $\alpha$ be a constant, $(\phi(\alpha,x;q,\bm{\omega}) ,\bm{r}(\alpha,x;q,\bm{\omega}))=(\alpha,\bm{0})$ $a.s.$, $\forall q\in[0,1]$;
\item Let $\alpha$ be a constant, $(\phi(\alpha+\Psi,x;q,\bm{\omega}) ,\bm{r}(\alpha+\Psi,x;q,\bm{\omega}))=(\alpha+\phi(\Psi,x;q',\bm{\omega}) ,\bm{r}(\Psi,x;q,\bm{\omega}))$ $a.s.$, $\forall q\in[0,1]$;
\item If $\Psi_1\leq \Psi_2\,a.s.$, then $\phi(\Psi_1,x;q,\bm{\omega})\leq \phi(\Psi_2,x;q,\bm{\omega})\,a.s.$
\item $\phi(\alpha\,\Psi_1+\beta \,\Psi_2,x;q,\bm{\omega})\leq \alpha\,\phi(\Psi_1,x;q,\bm{\omega})+\beta\,\phi(\Psi_2,x;q,\bm{\omega})\,a.s.$ for any constant $\alpha$ and $\beta$.
\end{itemize}
\end{proposition}
\begin{proof}
The first and the second properties are trivially proved by observing that $(\alpha,\bm{0})$ (resp. $(\alpha+\phi(\Psi,x;q,\bm{\omega}) ,\bm{r}(\Psi,x;q,\bm{\omega}))$ ) is actually a solution of the BSDE.

For the third property, let $(\phi_1,\bm{r}_1)$ and $(\phi_2,\bm{r}_2)$ be the solutions associated to $(\Psi_1,x)$ and $(\Psi_2,x)$ respectively, with $\Psi_1\leq \Psi_2$, then
\begin{multline}
\phi(\Psi_1,x;q,\bm{\omega})=\Gamma(\Psi_1,x,\bm{r}_1;q,\bm{\omega})\\
\leq \Gamma(\Psi_2,x,\bm{r}_1;q,\bm{\omega})\leq \Gamma(\Psi_2,x,\bm{r}_2;q,\bm{\omega})=\phi(\Psi_2,x;q,\bm{\omega}).
\end{multline}
The fourth relation can be proved in a similar way.
\end{proof}
The above results extends to the non-Markovian RSB the stochastic representation of the Parisi PDE (equation $\text{III.}55$ of chapter $\text{III}$ of \cite{VPM}), proposed by Chen and Auffinger (Theorem 3 in \cite{ChenAuf}). It is worth noting that Chen and Auffinger prove the variational representation of the Parisi formula, starting from the Parisi PDE and providing a stochastic representation.

Because non-Markovianity, we can not deal with a PDE, so the result is obtained by a completely different approach.

\section{Solution of Backward Stochastic Differential equations}
\label{sec:5}
This section, the existence of the solution of the stationary equation \eqref{selfEq1Aux} is proved. In the first paragraph, the solution is explicitly derived for a piecewise constant Parisi order parameter function $x$. Thence, in the second paragraph, the existence result is extended to any allowable Parisi order parameter by continuity. The results of the first section prove that the full-RSB-scheme provides a complete theory, that takes into account all the discrete-RSB solutions.

\subsection{The discrete-RSB solution}
We start by explicitly deriving the solution in the case where the Parisi parameter $x$ is a piecewise constant function:

A remarkable result of this paragraph is that, in this case, the free energy functional is equivalent to the one obtained in the discrete$-$RSB case (Eq. $22$ in \cite{MyPaper}).

Consider two increasing sequences of $K+2 \in \mathbb{N}$ numbers $q_0,\dots, q_{K+1}$ and $ x_0,\dots, x_{K+1}$ with
\begin{equation}
0=q_0\leq q_1\leq\dots\leq q_K\leq q_{K+1}=1
\end{equation}
and
\begin{equation}
0=x_0<x_1\leq\dots\leq x_K\leq x_{K+1}=1.
\end{equation}
The number $x_1$ must be non-zero.

The piecewise constant Parisi parameter function is constructed by such two sequences in such a way:
\begin{equation}
\label{discrete}
x(q)=\sum^{K+1}_{n=1} x_i \mathbb{1}_{(q_{i-1},q_{i}]}(q).
\end{equation}
We denote by $\chi^{\circ}\subset\chi$ the space of function of this form, for any $K\in \mathbb{N}$.

In this case, the right hand member of the BSDE \eqref{selfEq1Aux} is a sum of integrals defined on the intervals $(q_{i},q_{i+1}]$; in each interval, the Parisi parameter $x$ is a constant and it can be put outside the integral:
\begin{multline}
\label{equation_Discrete}
\phi\left(\,q\,,\bm{\omega}\,\right)=\Psi(q_{K+1},\bm{\omega})\\-
\sum^{K}_{i=n_q}\left(\int^{q_{i+1}}_{q_{i}\wedge q} \,\bm{r}(\,p\,, \bm{\omega}\,) \cdot \text{d}\bm{\omega}(p)-\frac{1}{2} x_{i+1}\int^{q_{i+1}}_{q_{i}\wedge q} dp\,\,\left \|\,\bm{r}(\,p, \,\bm{\omega}\,) \,\right\|^2\,\right),
\end{multline}
where $n_q$ is the integer number such as
\begin{equation}
q_{n_q}\leq q<q_{n_q+1}
\end{equation}
and
\begin{equation}
q_{i}\wedge q=\max \left\{\,q,\,q_i\,\right\}.
\end{equation}
We want to derive a self-consistency equation for the process $\phi$ that is equivalent to the BSDE \eqref{equation_Discrete}.

Let us consider the random variable $\zeta(\bm{r},x; 1,\Cdot|q)$, defined according to \eqref{zeta}:
\begin{multline}
\zeta(\bm{r},x; 1,\bm{\omega}|q)=\log \,\mathcal{E}(x\bm{r};1,\bm{\omega}|q)\\
=\sum^{K}_{i=n_q}\left(x_{i+1}\int^{q_{i+1}}_{q_{i}\wedge q} \,\bm{r}(\,p\,, \bm{\omega}\,) \cdot \text{d}\bm{\omega}(p)-\frac{1}{2} x^2_{i+1}\int^{q_{i+1}}_{q_{i}\wedge q} \text{d}p\,\,\left \|\,\bm{r}(\,p, \,\bm{\omega}\,) \,\right\|^2\,\right).
\end{multline}
Using the equation \eqref{equation_Discrete}, one finds:
\begin{equation}
\label{relation_r_phi}
\zeta(\bm{r},x;1,\bm{\omega}|q)=\sum^{K}_{i=n_q} x_{i+1}\left(\phi (\,q_{n+1}\,,\bm{\omega}\,)-\phi (\,q_{n}\wedge q\,,\bm{\omega}\,)\right),
\end{equation}
thus
\begin{equation}
\label{relation_G_phi}
\exp\left(x(q)\,\phi(q_{n_q+1},\bm{\omega})\right)=\exp\left(x(q)\,\phi(q,\bm{\omega})\right)\mathcal{E}\left(\,x\bm{r};q,\bm{\omega}|q_{n_q+1}\,\right).
\end{equation}
Let us assume, for now, that the DDE $\mathcal{E}(x\bm{r})$ is a true martingale. This assumption will be tested a posteriori. The martingale property implies that
\begin{equation}
\mathbb{E}\left[\mathcal{E}\left(\,x\bm{r};q,\bm{\omega}|q_{n_q+1}\,\right)\Big|\mathcal{F}_{q}\right]=1,
\end{equation}
then the self-consistency equation for the RSB value process $\phi$ can be derived by considering the following identity:
\begin{multline}
\label{discrete_equation_manipulation}
\phi\left(\,q\,,\bm{\omega}\,\right)=\frac{1}{x(q)} \log\, \exp \left(\,x(q) \phi\left(\,q\,,\bm{\omega}\,\right)\,\right)\\
=\frac{1}{x(q)} \log\left(\, \exp \left(\,x(q) \phi(\,q\,,\bm{\omega}\,)\,\right)\mathbb{E}\left[\mathcal{E}(\,x\bm{r};q,\bm{\omega}|q_{n_q+1}\,)\Big|\mathcal{F}_{q}\right]\,\right).
\end{multline}
By replacing the equality \eqref{relation_G_phi} in the above representation, we get:
\begin{equation}
\label{discrete_equation}
\phi\left(\,q\,,\bm{\omega}\,\right)=\frac{1}{x(q)} \log\,\mathbb{E}\left[\exp \left(\,x(q)\phi(\,q_{n_q+1}\,,\bm{\omega}\,)\,\,\right)\big|\mathcal{F}_{q}\right].
\end{equation}

The equation \eqref{discrete_equation} computed at the discontinuity points $0,\,q_1,\cdots,\,q_{K}$ leads to an iterative backward map that allows to derive progressively the $K+2$ random variables $\phi(\,q_{K},\bm{\omega}\,)$, $\cdots$, $\phi(\,0,\bm{\omega}\,)$ from the Wiener functional $\Psi(\,1,\bm{\omega}\,)$:
\begin{equation}
\label{iterative_discrete_cont}
\phi\left(\,q_n,\bm{\omega}\,\right)=\frac{1}{x_{n+1}} \log\,\mathbb{E}\left[\exp \left(\,x_{n+1}\phi(\,q_{n+1}\,,\bm{\omega}\,)\,\,\right)\big|\mathcal{F}_{q_n}\right].
\end{equation}
Note that the above iteration is equivalent to the discrete-RSB iteration given by equations $21$ and $22$ in \cite{MyPaper}.
\begin{proposition}
\label{inequality_prop_theo}
For any function $x\in \chi^{\circ}$, the process $\phi$, solution of the equation \eqref{discrete_equation}, is bounded, with
\begin{equation}
\label{inequality_prop}
\|\phi\|_{\infty,p}\leq\underset{\bm{\omega}\in\Omega}{\Max} |\Psi\left(\,1,\bm{\omega}\,\right)|\leq c.
\end{equation}
\end{proposition}
As a consequence $\phi\in S^p_{[0,1]}(\Omega)$, for any $p\geq 1$.
\begin{proof}
We start by proving the boundedness of the random variables $\phi(\,q_{K},\bm{\omega}\,)$, $\cdots$, $\phi(\,0,\bm{\omega}\,)$, by decreasing induction on $q_n$. For $q=q_{K+1}=1$, the Wiener functional $\phi(1, \Cdot )$ is bounded by \eqref{boundedness}:
\begin{equation}
\label{inequality_tot}
\underset{\bm{\omega}\in\Omega}{\Max} |\phi(1,\bm{\omega})|=\underset{\bm{\omega}\in\Omega}{\Max} |\Psi(\,1,\bm{\omega})|\leq c
\end{equation}
and using the decreasing induction hypothesis on $q_n$ we get
\begin{multline}
\label{bound}
\underset{\bm{\omega}\in\Omega}{\Max} |\phi\left(\,q_{n-1},\bm{\omega}\,\right)|= \underset{\bm{\omega}\in\Omega}{\max}\left|\frac{1}{x_{n}} \log\,\mathbb{E}\left[\exp \left(\,x_{n}\phi(\,q_{n}\,,\bm{\omega}\,)\,\,\right)\Big|\mathcal{F}_{q_{n-1}}\right]\right|\\
\leq \frac{1}{x_{n}} \log\,\mathbb{E}\left[\exp \left(\,x_{n}\underset{\bm{\omega}\in\Omega}{\Max}|\phi(\,q_{n}\,,\bm{\omega}\,)|\,\,\right)\Bigg|\mathcal{F}_{q_{n-1}}\right]\\= \underset{\bm{\omega}\in\Omega}{\Max}|\phi(\,q_{n}\,,\bm{\omega}\,)|,
\end{multline}
proving that
\begin{equation}
\label{inequality_tot}
\underset{\bm{\omega}\in\Omega}{\Max} |\phi\left(\,q_n,\bm{\omega}\,\right)|\leq \underset{\bm{\omega}\in\Omega}{\Max} |\phi\left(\,1,\bm{\omega}\,\right)|\leq c.
\end{equation}
Boundedness property trivially extends to the whole process $\phi$, for all $q\in(0,1]$, by equation \eqref{discrete_equation}.
\end{proof}
\begin{proposition}
\label{continuity_claim_prop}
For any function $x\in \chi^{\circ}$, the process $\phi$, solution of the equation \eqref{discrete_equation}, is an absolutely continuous functional with respect the claim. More specifically, given two claims $\Psi_1$ and $\Psi_2$ in $L_1^{\infty}(\Omega)$, let $\phi(\Psi_1,x)$ and $\phi(\Psi_2,x)$ the solution corresponding to the two claims, then the following inequality holds
\begin{equation}
\label{inequality_claim_prop}
\|\phi(\Psi_1,x)-\phi(\Psi_2,x)\|_{\infty,p}=\underset{\bm{\omega}\in\Omega}{\Max} |\Psi_1(\,1,\bm{\omega}\,)-\Psi_2(\,1,\bm{\omega}\,)|.
\end{equation}
\end{proposition}
\begin{proof}
Let
\begin{equation}
\delta\Psi=\Psi_1-\Psi_2,
\end{equation}
and
\begin{equation}
\Delta=\underset{\bm{\omega}\in\Omega}{\Max} |\Psi_1(\,1,\bm{\omega}\,)-\Psi_2(\,1,\bm{\omega}\,)|\,.
\end{equation}
Now, we introduce the normalized displacement $\delta\widehat{\Psi}$, given by:
\begin{equation}
\delta\widehat{\Psi}=\frac{1}{\Delta}(\Psi_1-\Psi_2)\,;
\end{equation}
Proposition \ref{Non:llinear_prop} implies that
\begin{equation}
\phi(\Psi_1,x)=\phi(\Psi_2+\Delta\delta\widehat{\Psi},x)\leq\phi(\Psi_2)+\Delta\phi(\delta\widehat{\Psi},x)
\end{equation}
and, since
\begin{equation}
\Psi_2= \Psi_1-\Delta\delta\widehat{\Psi}\leq \Psi_1+\Delta|\delta\widehat{\Psi}|\,,
\end{equation}
it follows that
\begin{equation}
\phi(\Psi_2,x)\leq \phi(\Psi_1+\Delta|\delta\widehat{\Psi}|,x)\leq\phi(\Psi_1)+\Delta\phi(|\delta\widehat{\Psi}|,x).
\end{equation}
Proposition \eqref{inequality_prop} yields
\begin{equation}
\underset{\bm{\omega}\in\Omega}{\Max} |\delta\widehat{\Psi}|=1\longrightarrow \|\phi( \delta\widehat{\Psi},x)\|_{\infty,p}\leq \|\phi( |\delta\widehat{\Psi}|,x)\|_{\infty,p}=1,
\end{equation}
then we conclude that
\begin{equation}
-\Delta\leq \phi(\Psi_1,x;q,\bm{\omega})-\phi(\Psi_2,x;q,\bm{\omega})\leq \Delta
\end{equation}
that ends the proof.
\end{proof}
Now, it remains to derive the vector process $\bm{r}$ (the auxiliary order parameter) that, together with the process $\phi$, verifies the equation \eqref{equation_Discrete}, and such as the DDE $\mathcal{E}(x\bm{r})$ is a martingale.

The auxiliary order parameter is derived as follows. In each interval $[q_n,q_{n+1}]$, with $0\leq n\leq K$, we define a process $J_n$ adapted to the filtration $\{\mathcal{F}_q,\,q\in[q_n,q_{n+1}]\}$, in such a way:
\begin{equation}
J_n(q,\bm{\omega})=\mathbb{E}\left[\exp \left(\,x_{n+1}\phi(\,q_{n+1}\,,\bm{\omega}\,)\,\,\right)\Big|\mathcal{F}_{q}\right],\,\,q_{n}\leq q\leq q_{n+1}.
\end{equation}
Since the Wiener functional $\phi(\,q_{n+1}\,,\cdot\,)$ is bounded, the process $J_n$ is a strictly positive bounded martingale for $q\in[q_n, q_{n+1}]$. By the martingale representation theorem for the Brownian motion, there exists a unique process $\bm{M}_n$ in $H^p_{[q_n, q_{n+1}]}(\Omega)$, for any $p\geq1$, such as
\begin{equation}
J_n(q,\bm{\omega})=J_n(q_{n},\bm{\omega})+\int^{q}_{q_{n}}\bm{M}_n(p,\bm{\omega})\cdot \text{d}\bm{\omega}(p),\,\,q_{n}\leq q\leq q_{n+1}.
\end{equation}
Since the process $J_n$ is strictly positive and continuous for all $q\in[q_{n},q_{n+1}]$, we can appy the It\^o formula to the process $\log J_n$, with the result that
\begin{multline}
\log J_{n_q}(q,\bm{\omega})=\\\log J_{n}(q_{n},\bm{\omega})+\int^q_{q_{n}}\frac{1}{J_n(p,\bm{\omega})}\bm{M}_n(p,\bm{\omega})\cdot \text{d}\bm{\omega}(p)-\frac{1}{2}\int^q_{q_{n}}\frac{1}{J^2_n(p,\bm{\omega})}\|\bm{M}(p,\bm{\omega})\|^2\,dp
\end{multline}
and thus we get
\begin{multline}
\label{proof_discrete}
\phi(q_{n_q+1},\bm{\omega})-\phi(q,\bm{\omega})=\frac{1}{x_{n_q}}\log J_{n_q}(q_{n_q},\bm{\omega})-\frac{1}{x_{n_q}}\log J_{n_q}(q,\bm{\omega})\\
=\int^{q_{n_q+1}}_{q}\frac{1}{x_{n_q+1} J_{n_q}(p,\bm{\omega})}\bm{M}_{n_q}(p,\bm{\omega})\cdot \text{d}\bm{\omega}(p)-\frac{1}{2}\int^{q_{n_q+1}}_{q}\frac{x_{n_q+1}}{x^2_{n_q+1}J^2_{n_q}(p,\bm{\omega})}\|\bm{M}_{n_q}(p,\bm{\omega})\|^2\,dp.
\end{multline}
Since $x_n>0$ for all $n>0$, then the integrals in the above equation are defined.
Put
\begin{equation}
\label{auxiliry_Order_piecewise}
\bm{r}(q,\bm{\omega})=\sum^K_{n=1} \frac{1}{x_i J_{i-1}(q,\bm{\omega})}\bm{M}_{i-1}(q,\bm{\omega})\mathbb{1}_{[q_{i-1},q_{i})}(q)
\end{equation}
then the pair $(\phi,\,\bm{r})$ satisfies the BSDE \eqref{equation_Discrete}. Moreover the process $\bm{r}$ verifies the following remarkable property, that will be crucial for the rest of the section.
\begin{proposition}
\label{trivial_bb}
For all $\bm{\omega}\in \Omega$, the process $\bm{r}$, given by \eqref{auxiliry_Order_piecewise}, verifies the following inequality
\begin{equation}
\label{bound_zet}
\left| \int^1_0x(q') \bm{r}(x'\,, \,\bm{\omega}) \cdot \text{d}\bm{\omega}(q') -\frac{1}{2}\int^1_0dq'\,x^2(q')\,\|\bm{r}(q'\,, \,\bm{\omega})\|^2\right|
\leq2 c\,
\end{equation}
that implies:
\begin{equation}
\label{bound_DDE}
e^{-2 c}\leq \mathcal{E}(x\bm{r};q,\bm{\omega}) \leq e^{2c},\quad \forall \bm{\omega}\in\Omega
\end{equation}
\end{proposition}
The above results implies that the DDE process $\mathcal{E}(x\bm{r})$ is a true martingale, so the vector process $\bm{r}$ is a proper solution of the auxiliary variational problem and identify an element of the domain set $ D_{[0,1]}(\Omega)$.
\begin{proof}
The proof of the inequality \eqref{bound_zet} is given by combining the equation \eqref{relation_r_phi} with the inequality \eqref{inequality_prop}, we obtain that
\begin{equation}
\begin{multlined}
\left| \int^1_0x(q') \bm{r}(q'\,, \,\bm{\omega}) \cdot \text{d}\bm{\omega}(q) -\frac{1}{2}\int^1_0dq'\,x^2(q')\,\|\bm{r}(q'\,, \,\bm{\omega})\|^2\right|
=|\zeta(\bm{r},x;1,\bm{\omega})|\\
\leq |\,\Psi(1,\bm{\omega})\,|+\sum^{K-1}_{i=0} (x_{i+1}-x_{i})|\phi (\,q_{i}\,,\bm{\omega})\,|\leq2 c.
\end{multlined}
\end{equation}
\end{proof}
It is worth noting that, since the processes $J_{i-1}$ and the function $x$ are strictly positive, and the process $\bm{M}_i$ is in $H_{[0,1]}^p(\Omega)$, then the process $\bm{r}$ is in $H_{[0,1]}^p(\Omega)$, for any $p\geq 1$. However, the boundedness of the process $\phi$ implies a stronger properties for the process $\bm{r}$ that is stated in the following proposition.

\begin{proposition}
\label{bound}
For every $p\in [0,\infty)$, there exist a universal constant $K_p$ such as for all the functions $x\in \chi^{\circ}$, the vector process $\bm{r}$, obtained by solving the equation \eqref{selfEq1Aux}, verifies:
\begin{equation}
\label{bound_rr}
\widetilde{\mathbb{E}}_{x\bm{r}}\left[\left(\int^{1}_{q}dq'\,\|\bm{r}(q',\bm{\omega})\|^{2}\right)^{\frac{p}{2}}\right]\leq K_p
\end{equation}
and
\begin{equation}
\label{bound_rr_2}
\mathbb{E}\left[\left(\int^{1}_{q}dq'\,\|\bm{r}(q',\bm{\omega})\|^{2}\right)^{\frac{p}{2}}\right]\leq e^{2c} K_p.
\end{equation}
\end{proposition}
We recall that the expectation value $\mathbb{E}[\Cdot]$ and $\widetilde{\mathbb{E}}_{x\bm{r}}[\Cdot]$ are evaluated by taking the value of the starting point of the Brownian motion $\bm{\omega}(0)$ fixed. The above inequalities are to be understood in an almost sure sense with respect the probability measure $\nu$.

The second inequality is a trivial consequence of the first one and Proposition \ref{trivial_bb}, indeed:
\begin{equation}
\mathbb{E}\left[\left(\int^{1}_{q}dq'\,\|\bm{r}(q',\bm{\omega})\|^{2}\right)^{\frac{p}{2}}\right]\leq e^{2c}\widetilde{\mathbb{E}}_{x\bm{r}}\left[\left(\int^{1}_{q}dq'\,\|\bm{r}(q',\bm{\omega})\|^{2}\right)^{\frac{p}{2}}\right].
\end{equation}
So, we just prove \eqref{bound_rr}.
\begin{proof}
Let $(\tau_n,n\in\mathbb{N})$ be the sequence of stopping times defined as follows
\begin{equation}
\tau_n=\sup \left\{q\in[0,1];\,\,\int^q_0 dq'\,x(q')^2\|\bm{r}(q',\bm{\omega})\|^2\leq n^2 \right\}
\end{equation}
and put $\inf\, \emptyset=1$. For each $n\in\mathbb{N}$, we set
\begin{equation}
\bm{r}_n(q,\bm{\omega})=\bm{r}(q,\bm{\omega})\theta(\tau_n-q),
\end{equation}
where the function $\theta$ is the Heaviside theta function. Since the stochastic integral $\int^1_0 dq\,x(q)^2\|\bm{r}(q,\bm{\omega})\|^2$ has a finite expectation, then $\tau_n\uparrow 1\,a.s.$. We start by proving the proposition for $\bm{r}_n$ and then we take the limit $n\to \infty$.

Let us define
\begin{equation}
\zeta_n(\alpha,\beta; \bm{\omega})=\,\alpha\, \int^{\tau_{n}}_{0}x(q)\bm{r}(q,\bm{\omega})\cdot \text{d}\bm{\omega}(q)- \frac{\beta\,}{2} \int^{\tau_{n}}_{0}dp x(q)^2\,\|\bm{r}(q,\bm{\omega})\|^2,
\end{equation}
where $\alpha$ and $\beta$ are two real numbers. We have
\begin{equation}
\zeta(\bm{r},x;\tau_n,\bm{\omega})=\zeta_n\left(\tfrac{1}{2},\tfrac{1}{4},\bm{\omega}\right)+\zeta_n\left(\tfrac{1}{2},\tfrac{3}{4},\bm{\omega}\right)\leq \zeta_n\left(\tfrac{1}{2},\tfrac{1}{4},\bm{\omega}\right)+\tfrac{1}{2}\zeta(\bm{r},x;\tau_n,\bm{\omega}).
\end{equation}
The definition of $\tau_n$ and of the process $\bm{r}_n$ implies that DDE $\mathcal{E}(x\bm{r}_n/2)$, is a true strictly positive martingale, that is
\begin{equation}
\mathbb{E}\left[\mathcal{E}\left(\tfrac{1}{2}x\bm{r}_n;1,\bm{\omega}\right)\,\right]=\mathbb{E}\left[e^{\zeta_n\left(\frac{1}{2},\frac{1}{4},\bm{\omega}\right)}\,\right]=1,
\end{equation}
so we can consider the Girsanov change of measure from the $n-$component Wiener measure $\mathbb{W}$ to the equivalent measure $\widetilde{\mathbb{W}}_{x\bm{r}_n/2}$.

As usual, the symbol $\widetilde{\mathbb{E}}_{x\bm{r}_n/2}[\Cdot]$ will denotes the expectation value with respect the measure $\widetilde{\mathbb{W}}_{x\bm{r}_n/2}$ and the process $\bm {W}_{x\bm{r}_n/2}$ is the $n-$components vector Brownian motion with respect the measure $\widetilde{\mathbb{W}}_{x\bm{r}_n/2}$:
\begin{equation}
\label{brownian_proofrr}
\bm{W}_{\bm{r}_n/2}(q,\bm{\omega})=\bm{\omega}(q)-\frac{1}{2}\int^q_0 \text{d}q' x(q') \bm{r}(q',\bm{\omega}).
\end{equation}
From a straightforward computation, we get
\begin{equation}
\widetilde{\mathbb{E}}_{x\bm{r}_n}\left[\left(\int^{\tau_n}_0 \text{d}q'\,\|\bm{r}(q',\bm{\omega})\|^{2}\,\right)^{\frac{p}{2}}\,\right]
\leq \widetilde{\mathbb{E}}_{\frac{1}{2}x\bm{r}_n}\left[e^{\frac{1}{2}\zeta(\bm{r},x;\tau_n,\bm{\omega})}\left(\int^{\tau_n}_0 \text{d}q'\,\|\bm{r}(q',\bm{\omega})\|^{2}\Bigg|\,\right)^{\frac{p}{2}}\,\right].
\end{equation}
and H\"{o}lder inequality for any $p\geq1$ yields
\begin{multline}
\label{holder_prrof}
\widetilde{\mathbb{E}}_{\frac{1}{2}x\bm{r}_n}\left[e^{\tfrac{1}{2}\zeta(\bm{r},x;\tau_n,\bm{\omega})}\left(\int^1_0 \text{d}q'\,\|\bm{r}_n(q',\bm{\omega})\|^{2}\Bigg|\,\right)^{\frac{p}{2}}\,\right]\\
\leq \widetilde{\mathbb{E}}_{\frac{1}{2}x\bm{r}_n}\left[\left(\int^1_0 \text{d}q'\,\|\bm{r}_n(q',\bm{\omega})\|^{2}\Bigg|\,\right)^{p}\,\right]^{\tfrac{1}{2}}\widetilde{\mathbb{E}}_{\frac{1}{2}x\bm{r}_n}\left[e^{\zeta(\bm{r},x;\tau_n,\bm{\omega})}\,\right]^{\tfrac{1}{2}}.
\end{multline}
and by Burkholder-Davis-Gundy inequality\cite{YoRev}, there exist a universal constant $C_p$, depending on $p$, such as:
\begin{multline}
\label{Burk}
\widetilde{\mathbb{E}}_{\frac{1}{2}x\bm{r}_n}\left[\left(\int^1_0 \text{d}q'\,\|\bm{r}_n(q',\bm{\omega})\|^{2}\Bigg|\,\right)^{p}\,\right]\\
\leq C_p \widetilde{\mathbb{E}}_{\frac{1}{2}x\bm{r}_n}\left[\left(\underset{q\in[0,1]}{\sup}\left|\int^1_0 d\bm{W}_{\frac{1}{2}\bm{r}_n}(q',\bm{\omega}) \,\bm{r}_n(q',\bm{\omega})\right|\,\right)^{2p}\,\right].
\end{multline}
By definition \eqref{brownian_proofrr} and the stationary equation \eqref{selfEq1Aux}, we have
\begin{multline}
\int^{q}_0 d\bm{W}_{\frac{1}{2}\bm{r}_n}(q',\bm{\omega}) \,\bm{r}_n(q',\bm{\omega})\\
=\int^q_0 \text{d}\bm{\omega}(q') \,\bm{r}_n(q',\bm{\omega})-\frac{1}{2}\int^q_0 \text{d}q' x(q') \,\|\bm{r}_n(q',\bm{\omega})\|^2\\=\phi(q)-\phi(0).
\end{multline}
By propositions \ref{inequality_prop} and \ref{bound_DDE}, the inequalities \eqref{holder_prrof} and \eqref{Burk} yield:
\begin{equation}
\widetilde{\mathbb{E}}_{x\bm{r}_n}\left[\left(\int^{\tau_n}_0 \text{d}q'\,\|\bm{r}(q',\bm{\omega})\|^{2}\,\right)^{\frac{p}{2}}\,\right]\leq \sqrt{C_p} e^{c}c^{p},\quad \forall n\in\mathbb{N}
\end{equation}
that proves the inequality \eqref{bound_rr} with $K_p= \sqrt{C_p} e^{c}c^{p}$.

The inequality \eqref{bound_rr_2} is an immediate consequence of the inequality \eqref{bound_DDE} and \eqref{bound_rr}.
\end{proof}
By proposition \ref{bound} the $H_{[0,1]}^p(\Omega)-$norm of the process $\bm{r}$ is dominated by a constant that dose not depends on the POP. This property will play a crucial role in the next paragraph.
\subsection{Extension to continuous POP}
In this paragraph we prove the existence of the solution of the BSDE \eqref{selfEq1Aux} when the Parisi order parameter is a generic increasing function $x\in \chi$.
The proof is quite technical and relies on several intermediate results.

Intuitively, we may proceed by approximating the POP through elements of $\chi^{\circ}$. We show that given a proper sequence of functions in $\chi^{\circ}$ that converges uniformly to a POP $x\in\chi$, the sequence of the solutions converges to a solution of the stationary equation \eqref{selfEq1Aux} corresponding to $x$.

To this aim, we need to study the dependence of the processes defined in \eqref{discrete_equation} and \eqref{auxiliry_Order_piecewise} on the corresponding POP. Let us denote by $(\phi(x),\bm{r}(x) )$ the solution of the BSDE \eqref{selfEq1Aux} corresponding to a given POP $x\in\chi^{\circ}$.

Note that, since the elements of $\chi^{\circ}$ are strictly positive functions, the map $\chi^{\circ}\ni x\mapsto \phi(x)\in S^p_{[0,1]}(\Omega)$ is continuous and infinitely differentiable. By contrast, a continuous POP $x$ may be arbitrary close to $0$ at $q\to 0$, so the extension of this property to the general case is not obvious.

The results in the next proposition allows to compare two process $\phi(x^{(1)})$ and $\phi(x^{(2)})$, corresponding to the piecewise constant POPs $x^{(0)}$ and $x^{(1)}$.
\begin{proposition}
\label{derivative_theo}
Let $(\phi^{(t)},\bm{r}^{(t)})$ be the solutions relating to the Parisi order parameters $x^{(t)}$ respectively. The next results
Consider two POPs $x^{(0)}$ and $x^{(1)}$ in $\chi^{\circ}$. Let
\begin{equation}
\delta x=x^{(1)}-x^{(0)},
\end{equation}
and consider
\begin{equation}
x^{(t)}=(1-t)x^{(0)}+t x^{(1)}\in \chi^{\circ}.
\end{equation}
Let $(\phi^{(t)},\bm{r}^{(t)})$ be the solution corresponding to the POP $x^{(t)}$. Then, for all $q\in [0,1]$ and $t\in [0,1]$ and almost all $\bm{\omega}\in \Omega$, the quantity $\phi^{(t)}(q,\bm{\omega})$ is derivable on $t$ and
\begin{equation}
\label{first_derivative}
\frac{\partial \phi^{(t)}(q,\bm{\omega})}{\partial t}=\frac{1}{2}\widetilde{\mathbb{E}}_{x^{(t)}\bm{r}^{(t)}}\left[\int^{1}_{q}\text{d}p\,\delta x(p)\,\|\bm{r}^{(t)}(p,\bm{\omega})\|^2\Bigg| \mathcal{F}_q\right],
\end{equation}
\end{proposition}
An immediate consequence of the above proposition is:
\begin{corollary}
Given two POPs $x^{(0)}$ and $x^{(1)}$ in $\chi^{\circ}$ such as
\begin{equation}
x^{(0)}(q)\leq x^{(1)}(q),\,\forall q\in[0,1]
\end{equation}
then
\begin{equation}
\phi(x^{(0)})\leq \phi(x^{(1)}).
\end{equation}
\end{corollary}
\begin{proof}[Proof of Proposition \ref{derivative_theo}]
Let $K$ be the number of discontinuity points $0= q_0<q_1< \cdots q_K<q_{K+1}=1$ of the function $x^{(t)}$.

We start by proving the formula of the first derivative. At $q=1$, the random variable $\phi^{(t)}(1,\Cdot)$ does not depends on $t$, that is
\begin{equation}
\label{recursion_derivative_start}
\frac{\partial}{\partial t}\phi^{(t)}(1,\bm{\omega})=\frac{\partial}{\partial t}\Psi(1,\bm{\omega})=0.
\end{equation}

For $q<1$, we proceed by differentiating the right member of the recursion \eqref{iterative_discrete_cont}. The chain rule yields a recursive equation for the derivative of $\phi^{(t)}$. For $q\in[q_n,q_{n+1}]$, with $0\leq n\leq K$, we have:
\begin{multline}
\label{iterative_der_proto}
\frac{\partial}{\partial t}{\phi^{(t)}}(q,\bm{\omega})=\frac{\partial}{\partial t}\left(\frac{1}{x^{(t)}(q)} \log\,\mathbb{E}\left[\exp \left(\,x^{(t)}(q)\phi^{(t)}(\,q_{n_q+1}\,,\bm{\omega}\,)\,\,\right)\bigg|\mathcal{F}_{q}\right]\right)\\
=\frac{\delta x(q)}{x^{(t)}(q)}\frac{\mathbb{E}\left[\exp \left(\,\phi^{(t)}(\,q_{n_q+1}\,,\bm{\omega}\,)\,\,\right)\phi^{(t)}(\,q_{n_q+1}\,,\bm{\omega}\,)\big|\mathcal{F}_{q}\right]}{\mathbb{E}\left[\exp \left(\,x^{(t)}(q)\phi^{(t)}(\,q_{n_q+1}\,,\bm{\omega}\,)\,\,\right) \big|\mathcal{F}_{q}\right]}\\-\frac{\delta x(q)}{(x^{(t)}(q))^2} \log\,\mathbb{E}\left[\exp \left(\,x^{(t)}(q)\phi^{(t)}(\,q_{n_q+1}\,,\bm{\omega}\,)\,\,\right)\big|\mathcal{F}_{q}\right]\\
+\frac{\mathbb{E}\left[\exp \left(\,x^{(t)}(q)\phi^{(t)}(\,q_{n_q+1}\,,\bm{\omega}\,)\,\,\right)\tfrac{\partial}{\partial t}\phi^{(t)}(\,q_{n_q+1}\,,\bm{\omega}\,)\big|\mathcal{F}_{q}\right]}{\mathbb{E}\left[\exp \left(\,x^{(t)}(q)\phi^{(t)}(\,q_{n_q+1}\,,\bm{\omega}\,)\,\,\right) \big|\mathcal{F}_{q}\right]}.
\end{multline}
Now, the equation \eqref{iterative_discrete_cont} implies
\begin{equation}
\mathbb{E}\left[\exp \left(\,x^{(t)}(q)\phi^{(t)}(\,q_{n_q+1}\,,\bm{\omega}\,)\,\,\right) \bigg|\mathcal{F}_{q}\right]=\exp\left(x^{(t)}_n\phi^{(t)}(\,q_{n_q+1}\,,\bm{\omega}\,)\right )
\end{equation}
and
\begin{equation}
\frac{\delta x(q)}{(x^{(t)}(q))^2} \log\,\mathbb{E}\left[\exp \left(\,x^{(t)}(q)\phi^{(t)}(\,q_{n_q+1}\,,\bm{\omega}\,)\,\,\right)\bigg|\mathcal{F}_{q}\right]=\frac{\delta x(q)}{x^{(t)}(q)}\phi^{(t)}\left(\,q\,,\bm{\omega}\,\right)
\end{equation}
and, since $\phi(q,\Cdot)$ is $\mathcal{F}_q$ measurable, we have the following identity:
\begin{equation}
\phi(q,\bm{\omega})
=\widetilde{\mathbb{E}}_{x^{(t)}\bm{r}^{(t)}}\big[\phi^{(t)}(q,\bm{\omega})\big| \mathcal{F}_q\big].
\end{equation}
By replacing the above three relations in the equation \eqref{iterative_der_proto}, we finally get
\begin{equation}
\label{iterative_der_0}
\frac{\partial \phi^{(t)}(q,\bm{\omega})}{\partial t}=
\frac{\delta x_{n_q}}{x_{n_q}^t}\widetilde{\mathbb{E}}_{x^{(t)}\bm{r}^{(t)}}\big[\phi^{(t)}(q_{n+1},\bm{\omega})-\phi^{(t)}(q,\bm{\omega})\big| \mathcal{F}_q\big]+\widetilde{\mathbb{E}}_{x^{(t)}\bm{r}^{(t)}}\left[\frac{\partial \phi^{(t)}(q_{n+1},\bm{\omega})}{\partial t}\Bigg| \mathcal{F}_q\right].
\end{equation}
The solution of the above recursive equation, together with the starting condition \eqref{recursion_derivative_start}, is
\begin{equation}
\label{der_1}
\frac{\partial \phi^{(t)}(q,\bm{\omega})}{\partial t}=\sum^K_{n=n_q}\frac{\delta x_{n+1}}{x_{n+1}^t}\widetilde{\mathbb{E}}_{x^{(t)}\bm{r}^{(t)}}\big[\phi^{(t)}(q_{n+1},\bm{\omega})-\phi^{(t)}(q_n\wedge q,\bm{\omega})\big| \mathcal{F}_{q}\big].
\end{equation}
Substituting the process $\phi$ with the stationary equation for discrete Parisi order parameter \eqref{equation_Discrete}, one finds:
\begin{equation}
\begin{multlined}
\frac{\delta x_{n}}{x_{n}^t}\widetilde{\mathbb{E}}_{x^{(t)}\bm{r}^{(t)}}\big[\phi^{(t)}(q_{n+1},\bm{\omega})-\phi^{(t)}(q_n,\bm{\omega})\big| \mathcal{F}_{q_{n_q}}\big]\\
=\frac{\delta x_{n}}{x_{n}^t}\widetilde{\mathbb{E}}_{x^{(t)}\bm{r}^{(t)}}\Bigg[\int^{q_{n+1}}_{q_n} \text{d}\bm{\omega}(q') \cdot\bm{r}(q',\bm{\omega})-\frac{x_n^{(t)}}{2}\int^{q_{n+1}}_{q_n} \text{d}q' \,\|\bm{r}(q',\bm{\omega})\|^2\Bigg| \mathcal{F}_{q_{n_q}}\Bigg]\\
=\frac{\delta x_{n}}{2}\widetilde{\mathbb{E}}_{x^{(t)}\bm{r}^{(t)}}\left[\int^{q_{n+1}}_{q_n} \text{d}q' \,\|\bm{r}(q',\bm{\omega})\|^2\Bigg| \mathcal{F}_{q_{n_q}}\right],
\end{multlined}
\end{equation}
that proves \eqref{first_derivative}.
\end{proof}
From the above results and the proposition \ref{bound}, we deuce that the process $\partial_t \phi^{(t)}$ is almost surely bounded.

Now, we state the most remarkable property of the map $\chi^{\circ}\ni x\mapsto (\phi(x),\bm{r}(x)\,)\in S^p_{[0,1]}(\Omega)$
\begin{theorem}
\label{uniform_convergence_theo}
Let $x^{(1)}$ and $x^{(2)}$ be two elements of $\chi^{\circ}$. Then, for any $p>1$ there exist a constant $K_p$ depending only $p$ such that:
\begin{equation}
\label{uniform_converge_phi}
\mathbb{E}_{\mathbb{W}}\left[\underset{q\in [0,1]}{\sup}\Big|\phi(x^{(2)};q,\bm{\omega})-\phi(x^{(1)};q,\bm{\omega})\Big|^p\right]^{\frac{1}{p}}\leq K_p \|x^{(2)}-x^{(1)}\|_{\infty}
\end{equation}
and
\begin{equation}
\mathbb{E}_{\mathbb{W}}\left[\left(\int^1_0\left\|\bm{r}(x^{(2)};q,\bm{\omega})-\bm{r}(x^{(1)};q,\bm{\omega})\right\|\right)^p\right]^{\frac{1}{p}}\leq K_p \|x^{(2)}-x^{(1)}\|_{\infty}.
\end{equation}
\end{theorem}
This implies, in particular, that if two POP $x^{(1)}$ and $x^{(2)}$ are "close to each other”, then the "level of approximation” of the solutions $(\phi(x^{(1)},\bm{r}(x^{(1)})$ provided by the solution $(\phi(x^{(2)},\bm{r}(x^{(2)})$ depends only by the $\|\cdot\|_{\infty}-$distance between the two POPs. This result is very important. In fact, any POP in $\chi$ is arbitrary close to a POP in $\chi^{\circ}$.
\begin{proof}
The inequality \eqref{uniform_convergence_theo} is an immediate consequence of the equation \eqref{first_derivative} and the proposition \ref{bound}.
Put
\begin{equation}
\delta x=x^{(2)}-x^{(1)},\quad \delta\phi=\phi(x^{(2)})-\phi(x^{(1)}),\quad \delta\bm{r}=\bm{r}(x^{(2)})-\bm{r}(x^{(1)}).
\end{equation}
Since $\phi(x^{(1)})$ and $\phi(x^{(2)})$ are bounded and the processes $\bm{r}(x^{(1)})$ and $\bm{r}(x^{(2)})$ are in $H^p_{[0,1]}(\Omega)$, then the process $\delta \phi$ is bounded and $\delta\bm{r}$ is in $H^p_{[0,1]}(\Omega)$.

We use the same notation of theorem \ref{derivative_theo}. Let
\begin{equation}
x^{(t)}(q)=t x^{(1)}(q)+(1-t)x^{(2)}(q)\in \chi^{\circ}, \quad t\in [0,1].
\end{equation}
By theorem \ref{derivative_theo}, the process $\phi^{(t)}$ is derivable over $t$, that implies
\begin{equation}
\delta \phi(q,\bm{\omega})=\int^1_0 dt \frac{\partial \phi^{(t)}(q,\bm{\omega})}{\partial t}= \frac{1}{2}\int^1_0 dt\,\widetilde{\mathbb{E}}_{x^{(t)}\bm{r}^{(t)}}\left[\int^{1}_{q}\text{d}p\,\delta x(p)\,\|\bm{r}^{(t)}(p,\bm{\omega})\|^2\Bigg| \mathcal{F}_q\right]
\end{equation}
from which it follows that
\begin{equation}
\label{obvious_inequality}
|\delta \phi(q,\bm{\omega})|\leq\frac{1}{2}\int^1_0 dt\,\widetilde{\mathbb{E}}_{x^{(t)}\bm{r}^{(t)}}\left[\int^{1}_{0}\text{d}p\,\delta x(q')\,\|\bm{r}^{(t)}(q',\bm{\omega})\|^2\Bigg| \mathcal{F}_q\right]
\end{equation}
Since the process $\delta\bm{r}$ is in $H^p_{[0,1]}(\Omega)$, then the process integrated over $t$ in the right-hand side of the above inequality is a non-negative martingale bounded in $L_p$ with respect the probability measure $\widetilde{\mathbb{W}}_{\bm{r}^{(t)}}$, for all $t\in [0,1]$. As a consequence, Doob inequality and the proposition \ref{bound} yield
\begin{equation}
\begin{multlined}
\widetilde{\mathbb{E}}_{x^{(t)}\bm{r}^{(t)}}\left[\left(\underset{q\in[0,1]}{\sup}\widetilde{\mathbb{E}}_{x^{(t)}\bm{r}^{(t)}}\left[\int^{1}_{0}\text{d}p\,\delta x(q')\,\|\bm{r}^{(t)}(q',\bm{\omega})\|^2\Bigg| \mathcal{F}_q\right]\right)^p\right]\\\leq\left(\frac{p}{p-1}\right)^p\widetilde{\mathbb{E}}_{x^{(t)}\bm{r}^{(t)}}\left[\left(\,\int^{1}_{0}\text{d}q'\,\delta x(q')\,\|\bm{r}^{(t)}(q',\bm{\omega})\|^2\right)^p\right]\leq\left(\frac{p}{p-1}\right)^p(K_p\|\delta x\|_{\infty})^p,\\\quad \forall t\in[0,1]\,\,\,\text{and}\,\,\, p>1
\end{multlined}
\end{equation}
Combining the above result with the inequalities \eqref{obvious_inequality} and the proposition \ref{bound}, and since the process $\delta \phi$ is bounded, we finally get:
\begin{equation}
\begin{multlined}
\mathbb{E}_{\mathbb{W}}\left[\underset{q\in[0,1]}{\sup}|\delta \phi(q,\bm{\omega})|^p\right]\leq \,\mathbb{E}_{\mathbb{W}}\left[\underset{q\in[0,1]}{\sup}\left(\int^1_0 dt\,\frac{\partial \phi^{(t)}(q,\bm{\omega})}{\partial t}\right)^p\right]\\
\leq e^{2 c}\left(\frac{p}{2(p-1)}\right)^p\underset{t\in [0,1]}{\sup}\,\widetilde{\mathbb{E}}_{x^{(t)}\bm{r}^{(t)}}\left[\left(\,\int^{1}_{0}\text{d}q'\,\delta x(q')\,\|\bm{r}^{(t)}(q',\bm{\omega})\|^2\right)^p\right]\\\leq \left(\frac{p\,K_p}{2(p-1)}\right)^pe^{2 c}\|\delta x\|^p_{\infty},\,\forall p>1
\end{multlined}
\end{equation}
that proves the inequality \eqref{uniform_converge_phi}, with $a_p=e^{2 c/p}p\,K_p/(2p-2)$.

Now we prove that the process $\delta\bm{r}$ has the same bound. At $q=1$, the process $\delta \phi$ verifies:
\begin{equation}
\delta\phi(1,\bm{\omega})=0,
\end{equation}
and from the two auxiliary stationary equations associated to the POPs $x^{(1)}$ and to $x^{(2)}$ and the above relation, one gets
\begin{equation}
\begin{multlined}
\label{equation_delta}
0=\delta \phi(1,\bm{\omega})
=\delta\phi(0,\bm{\omega})+\int^1_0 \delta \bm{r}(q,\bm{\omega}) \cdot \text{d}\bm{\omega}(q)-\frac{1}{2}\int^1_0 \text{d}q\,\delta x(q)\,\|\bm{r}_{2}(q,\bm{\omega})\|^2\\-\frac{1}{2}\int^1_0 \text{d}q\,x^{(1)}(q)\left(\bm{r}_{1}(q,\bm{\omega})+\bm{r}_{2}(q,\bm{\omega})\right)\cdot\delta \bm{r}(q,\bm{\omega}).
\end{multlined}
\end{equation}
By applying the It\^o formula to $(\delta \phi(1,\bm{\omega}))^2$, it follows that
\begin{equation}
\begin{multlined}
(\delta\phi(0,\bm{\omega}))^2+2\int^1_0\delta\phi(q,\bm{\omega})\bm{r}(q,\bm{\omega}) \cdot \text{d}\bm{\omega}(q)
-\int^1_0 \text{d}q\,\delta x(q)\,\delta\phi(q,\bm{\omega})\|\bm{r}_{2}(q,\bm{\omega})\|^2\\-\int^1_0 \text{d}q\,x^{(1)}(q)\delta\phi(q,\bm{\omega})\left(\,\bm{r}_{1}(q,\bm{\omega})+\bm{r}_{2}(q,\bm{\omega})\,\right)\cdot\delta \bm{r}(q,\bm{\omega})+\int^1_0 \text{d}q\|\delta\bm{r}_{2}(q,\bm{\omega})\|^2=0.
\end{multlined}
\end{equation}
All the quantities in the above expression are in $L_{[0,1]}^p(\Omega)$. We put $\int^1_0 \text{d}q\|\delta\bm{r}_{2}(q,\bm{\omega})\|^2$ on the left-hand side of the equation and the other terms in the right-hand side and take the absolute value raised to the power $p$ of both side. Using the inequality $|A+B+C+D|^p\leq4^{p-1}(|A|^p+|B|^p+|C|^p+|D|^{p})$ and taking the expectation value, we gets
\begin{equation}
\label{inequality_r}
\mathbb{E}_{\mathbb{W}}\left[\left(\int^1_0 \text{d}q\|\delta\bm{r}_{2}(q,\bm{\omega})\|^2\right)^p\right]\leq\text{I}+\text{II}+\text{III}+\text{IV}
\end{equation}
where
\begin{equation}
\text{I}=4^{p-1}\mathbb{E}_{\mathbb{W}}\left[\left|\delta\phi(0,\bm{\omega})\right|^{2p}\right]\leq 4^{p-1}a^{2p}_{2p}\|\delta x\|^{2p}_{\infty},
\end{equation}
\begin{equation}
\begin{multlined}
\text{II}=2*8^{p-1}\mathbb{E}_{\mathbb{W}}\left[\left|\int^1_0\delta\phi(q,\bm{\omega})\delta\bm{r}(q,\bm{\omega}) \cdot \text{d}\bm{\omega}(q)\right|^p\right]\\\leq 2^{\frac{5}{2}p-2}p^{p-1}(p-1)\mathbb{E}_{\mathbb{W}}\left[\left(\int^1_0\text{d}q \left(\delta\phi(q,\bm{\omega})\right)^{2}\|\delta\bm{r}(q,\bm{\omega})\|^2\right)^{\frac{p}{2}}\right]\\
\leq 2^{\frac{5}{2}p-2}p^{p-1}(p-1)\mathbb{E}_{\mathbb{W}}\left[\left(\underset{q\in[0,1]}{\sup}|\delta \phi(q,\bm{\omega})|^{p}\right)\left(\int^1_0\text{d}q \|\delta\bm{r}(q,\bm{\omega})\|^2\right)^{\frac{p}{2}}\right]\\
\leq 2^{\frac{5}{2}p-2}p^{p-1}(p-1)a_{2p}^{p}\,\|\delta x\|_{\infty}^{p} \mathbb{E}_{\mathbb{W}}\left[\left(\int^1_0\text{d}q \|\delta\bm{r}(q,\bm{\omega})\|^2\right)^{p}\right]^{1/2},
\end{multlined}
\end{equation}
\begin{equation}
\begin{multlined}
\text{III}=4^{p-1}\mathbb{E}_{\mathbb{W}}\left[\left|\int^1_0 \text{d}q\,\delta x(q)\,\delta\phi(q,\bm{\omega})\|\bm{r}_{2}(q,\bm{\omega})\|^2\right|^{p}\right]\\
\leq 4^{p-1}\mathbb{E}_{\mathbb{W}}\left[\left(\underset{q\in[0,1]}{\sup}|\delta \phi(q,\bm{\omega})|^{p}\right)\left(\int^1_0 \text{d}q\,|\delta x(q)|\|\bm{r}_{2}(q,\bm{\omega})\|^2\right)^{p}\right]\\
\leq 4^{p-1}\mathbb{E}_{\mathbb{W}}\left[\underset{q\in[0,1]}{\sup}|\delta \phi(q,\bm{\omega})|^{2p}\right]^{\frac{1}{2}}\mathbb{E}_{\mathbb{W}}\left[\left(\int^1_0 \text{d}q\,|\delta x(q)|\|\bm{r}_{2}(q,\bm{\omega})\|^2\right)^{2p}\right]^{\frac{1}{2}}\\
\leq 4^{p-1}a^p_{2p}\,e^c K^p_{2p}\|\delta x\|_{\infty}^{2p}
\end{multlined}
\end{equation}
\begin{equation}
\begin{multlined}
\text{IV}=4^{p-1}\mathbb{E}_{\mathbb{W}}\left[\left|\int^1_0 \text{d}q\,x^{(1)}(q)\delta\phi(q,\bm{\omega})\left(\,\bm{r}_{1}(q,\bm{\omega})+\bm{r}_{2}(q,\bm{\omega})\,\right)\cdot\delta \bm{r}(q,\bm{\omega})\right|^{p}\right]\\
\leq 4^{p-1}\mathbb{E}_{\mathbb{W}}\left[\underset{q\in[0,1]}{\sup}|\delta \phi(q,\bm{\omega})|^{p}\left|\int^1_0 \text{d}q\,x^{(1)}(q)\left(\,\bm{r}_{1}(q,\bm{\omega})+\bm{r}_{2}(q,\bm{\omega})\,\right)\cdot\delta \bm{r}(q,\bm{\omega})\right|^{p}\right]\\
\leq 4^{p-1}e^{\frac{ c}{2}}K^{p/2}_{2p}a_{4p}^{p}\|\delta x\|^{p}_{\infty}\mathbb{E}_{\mathbb{W}}\left[\left(\int^1_0 \text{d}q\,\|\delta \bm{r}(q,\bm{\omega})\|^2\right)^{p}\right]^{1/2}
\end{multlined}
\end{equation}
If we set
\begin{equation}
X=\mathbb{E}_{\mathbb{W}}\left[\left(\int^1_0\text{d}q \|\delta\bm{r}(q,\bm{\omega})\|^2\right)^{p}\right]^{1/2}
\end{equation}
and combine the above inequalities in \label{inequality_r}, the inequality \label{inequality_r} has the form:
\begin{equation}
\label{inequality_2}
X^2\leq \alpha_p \|\delta x\|_{\infty}^p X+\beta_p\|\delta x\|_{\infty}^{2p}
\end{equation}
where $\alpha_p$ and $\beta_p$ are two positive constants that depends only on $p$. That implies that
\begin{equation}
X\leq \,\frac{1}{2}\left(\alpha_p+\sqrt{\alpha^2_p+4\beta_p\,}\right)\,\|\delta x\|_{\infty}^{p}.
\end{equation}
and the proof is ended.
\end{proof}
We now present the main result of this paragraph.
\begin{theorem}
\label{convergence_result}
Given a POP $x\in\chi$, consider a sequence of piecewise constant POPs $(x^{(k)})\subset \chi^{\circ}$, where
\begin{equation}
\|x^{(k)}-x\|_{\infty}\leq 2^{-k}
\end{equation}
The sequence of the solutions $\left(\,(\phi_{x^{(k)}},\bm{r}_{x^{(k)}})\,\right)$ converges almost surely and in $H^p_{[0,1]}(\Omega) \times H^p_{[0,1]}(\Omega)$ norm to a pair $(\phi,\bm{r})$ that is a solution of the auxiliary stationary equation \eqref{selfEq1Aux} corresponding to the POP $x$.
\end{theorem}

\begin{proof}
Let us consider the sequence of pairs of non-negative random variables $\left(\,(U_{k},V_{k})\,\right)$, where
\begin{equation}
U_k=\underset{q\in[0,1]}{\sup}\left|\phi_{x^{(k+1)}}(q,\bm{\omega})-\phi_{x^{(k)}}(q,\bm{\omega})\right|
\end{equation}
and
\begin{equation}
V_k=\int^1_0 \text{d}q\|\bm{r}_{x^{(k+1)}}(q,\bm{\omega})-\bm{r}_{x^{(k)}}(q,\bm{\omega})\|^2.
\end{equation}
Theorem \eqref{uniform_convergence_theo} yields:
\begin{equation}
\mathbb{E}_{\mathbb{W}}\big[ U^p_k\big]\leq a_p2^{-p k},\quad
\mathbb{E}_{\mathbb{W}}\big[V^p_k\big]\leq b_p2^{-p k},
\end{equation}
consequently
\begin{equation}
\label{inequality_start}
\sum^{\infty}_{k=1}\mathbb{E}_{\mathbb{W}}\big[ U^p_k\big]\leq \infty,\quad \sum^{\infty}_{k=1}\mathbb{E}_{\mathbb{W}}\big[V^p_k\big]\leq \infty.
\end{equation}
from which it is straightforward to obtain that the sequence $\left(\,(\phi_{x^{(k)}},\bm{r}_{x^{(k)}})\,\right)$ converges in $S^p_{[0,1]}(\Omega) \times H^p_{[0,1]}(\Omega)$ to a pair $(\phi,\bm{r})$. Moreover, by Markov inequality and \eqref{inequality_start}, one gets
\begin{equation}
\mathbb{W}\left[\left\{\bm{\omega}\in\Omega;\,U_k>\epsilon \right\}\right]\leq \frac{\mathbb{E}\big[U^p_k\big]}{\epsilon^p}\leq a_p\left(\frac{1}{2^k\epsilon}\right)^p,
\end{equation}
and in the same way:
\begin{equation}
\mathbb{W}\left[\left\{\bm{\omega}\in\Omega;\,V_k>\epsilon \right\}\right]\leq b_p\left(\frac{1}{2^k\epsilon}\right)^p,
\end{equation}
where $\mathbb{W}[A]$ denotes the probability that the event $A$ occurs, according to the probability measure $\mathbb{W}$.

By Borel-Cantelli lemma \cite{Billingsley}, the above two Markov inequalities together with the convergence results in \eqref{inequality_start} imply that the sequence $\left((U^p_k,V_k^p)\right)$ converges almost surely to $(0,0)$ with respect the probability measure $\mathbb{W}$. In particular that implies that the sequence $\left(\,(\phi_{x^{(k)}},\bm{r}_{x^{(k)}})\,\right)$ converges almost surely to the pair $(\phi,\bm{r})$. Moreover, the definition of $(U_k)$ implies that $(\phi_{x^{(k)}})$ converges to $\phi$ almost surely uniformly in the interval $[0,1]$, so the process $\phi$ is continuous.

It remains to show that the pair $(\phi,\bm{r})$ is a solution of the equation \eqref{selfEq1Aux} corresponding to the POP $x$. Since the process $\bm{r}$ is in $ H^p_{[0,1]}(\Omega)$, for any $p\geq1$, then the It\^o integral $\int^1_q \text{d}\bm{\omega}(q')\cdot \bm{r}(q',\bm{\omega})$ and the integral $\int^1_q \text{d}q' x(q') \|\bm{r}(q',\bm{\omega})\|^2$ exist and are in $S_{[0,1]}^{q}(\Omega)$, for any $q\geq1$. Let
\begin{equation}
\text{I}_k(q;\bm{\omega})=\int^1_q \text{d}\bm{\omega}(q')\cdot \bm{r}_{x^{(k)}}(q',\bm{\omega}),\quad \text{II}_k(q,\bm{\omega})=\int^1_q \text{d}q' x^{(k)}(q') \|\bm{r}_{x^{(k)}}(q',\bm{\omega})\|^2
\end{equation}
and
\begin{equation}
\text{I}(q;\bm{\omega})=\int^1_q \text{d}\bm{\omega}(q')\cdot \bm{r}(q',\bm{\omega}),\quad \text{II}(q,\bm{\omega})=\int^1_q \text{d}q' x(q') \|\bm{r}(q',\bm{\omega})\|^2.
\end{equation}
Note that $\text{I}_k$, $\text{II}_k$, $\text{I}$ and $\text{II}$ are not adapted process, so we use the notation $(q;\bm{\omega})$ instead of $(q,\bm{\omega})$. We must prove the almost sure convergence of the sequences $(\text{I}_k)$ and $(\text{II}_k)$ to $\text{I}$ and $\text{II}$ respectively.

Let us define the following non-negative random variables
\begin{equation}
G_k=\underset{q\in[0,1]}{\sup}\left|\,\text{I}_k(q;\bm{\omega})-\text{I}(q;\bm{\omega})\,\right|=\underset{q\in[0,1]}{\sup}\left|\int^1_q \text{d}\bm{\omega}(q')\cdot \left(\bm{r}_{x^{(k)}}(q',\bm{\omega})-\bm{r}(q',\bm{\omega})\right)\right|
\end{equation}
and
\begin{equation}
\begin{multlined}
F_k=\underset{q\in[0,1]}{\sup}\left|\,\text{II}_k(q;\bm{\omega})-\text{II}(q;\bm{\omega})\,\right|\\=\underset{q\in[0,1]}{\sup}\left|\int^1_q \text{d}q'\left(x^{(k)}(q')\|\bm{r}_{x^{(k)}}(q',\bm{\omega})\|^2-x(q') \|\bm{r}(q',\bm{\omega})\|^2\right)\right|.
\end{multlined}
\end{equation}
By BDG inequality, there is a positive constant $C_p$, depending only on $p$, such as
\begin{equation}
\mathbb{E}_{\mathbb{W}}\left[G_k^p\right]\leq C_p \mathbb{E}_{\mathbb{W}}\left[V_k^{p/2}\right]\leq C_p b_p 2^{-kp}.
\end{equation}
Moreover, the inequality \eqref{bound_rr} yields
\begin{equation}
\begin{multlined}
\mathbb{E}_{\mathbb{W}}\left[F_k^p\right]\\\leq \|x^{(k)}-x\|^{p}_{\infty}\mathbb{E}_{\mathbb{W}}\left[\left(\int^1_q \text{d}q' \|\bm{r}^{(k)}(q',\bm{\omega})\|^2\right)^p\right]\\+\mathbb{E}_{\mathbb{W}}\left[\left(\int^1_q \text{d}q' x(q') \|\bm{r}^{(k)}(q',\bm{\omega})-\bm{r}(q',\bm{\omega})\|^2\right)^p\right]\\
\leq 2^{-k} e^{2c}(2c)^2+\mathbb{E}_{\mathbb{W}}\left[V_k^{p/2}\right]\leq 2 ^{-k} c_p
\end{multlined}
\end{equation}
where $c_p$ is a positive constant depending only on $p$. As for the sequences $(U_k)$ and $(V_k)$, the above two inequalities imply that the sequences $(G_k)$ and $(F_k)$ converge in $L^p(\mathbb{W},\Omega)$ and almost surely to $0$, so the random variables $(\text{I}_k)$ and $(\text{II}_k)$ converge almost surely uniformly in $q\in [0,1]$ to $\text{I}$ and $\text{II}$ respectively, and the proof is ended.
\end{proof}
Finally, we end this section with the following obvious, but important result
\begin{theorem}
\label{extending_theo}
The propositions \ref{inequality_prop_theo}, \ref{bound}, \ref{derivative_theo} and \ref{uniform_convergence_theo} hold for all the allowable POP $x\in \chi$. In particular, we have
\begin{equation}
\label{derivative_formula_on_X}
\frac{\delta\phi(x;\,q',\bm{\omega}) }{\delta x(q)}=\theta(q-q')\widetilde{\mathbb{E}}_{x\bm{r}}\left[ \|\bm{r}(q,\bm{\omega})\|^2\big|\mathcal{F}_{q'}\right].
\end{equation}
\end{theorem}

This result is a straightforward consequence of the convergence result of Theorem \ref{convergence_result}.

By corollary \ref{GlobalMinimumC}, the above theorem provides some important properties on the dependence of the Non-Markov RSB expectation $\Sigma(\Psi,x)$ defined in \eqref{definitions}. 
\section{Conclusion}
This manuscript expands the results presented in \cite{MyPaper} on the full-RSB variational free energy of the Ising spin glass on random regular graphs. In particular, it provides a detailed mathematical analysis of the so-called auxiliary variational problem.

We prove that the solution of the auxiliary variational problem is the solution of a proper backward stochastic differential equation. Finally, we provide the existence and uniqueness results for such an equation.

Remarkably, for a proper choice of the physical order parameters, the full-RSB free energy is equivalent to the discrete-RSB free energy, presented in \cite{MyPaper}.

\end{document}